\documentclass[journal]{IEEEtran}
\usepackage{graphicx}
\usepackage{amsmath}
\usepackage[noend]{algpseudocode}
\usepackage{algorithmicx,algorithm}
\usepackage{amsthm}
\usepackage{amsfonts}
\usepackage{subfigure}
\usepackage{color}
\usepackage{cite}
\usepackage{setspace}
\newtheorem{Lemma}{Lemma}
\newtheorem{Theorem}{Theorem}

\newtheorem{Definition}{Definition}
\newtheorem{Proposition}{Proposition}
\usepackage{amssymb}
\usepackage{color}
\graphicspath{{images/}}
\usepackage{booktabs}
\usepackage{multirow}
\usepackage{makecell}
\usepackage[numbers,sort&compress]{natbib}
\newtheorem{remark}{Remark}
\usepackage{stfloats}
\usepackage{amsmath}
\usepackage{booktabs}
\usepackage{amssymb}
\usepackage{enumitem}
\usepackage{threeparttable}
\usepackage{comment}
\usepackage{bm}
\begin{document}
	
	\title{Digital Tides: A Fluid-Dynamic Framework for Flux-Aware Infrastructure Provisioning in UAV Logistics Networks}
	
\author{Wen-Yu Dong, Song Zhao, Rui-Si Han, Qi Bi,~\IEEEmembership{Fellow,~IEEE}, Sheng Chen,~\IEEEmembership{Life Fellow,~IEEE}	%

	\thanks{W.-Y. Dong, S. Zhao and Q. Bi are with Future Technology Research Center, China Telecom Research Institute, Beijing 102209, China (E-mails: dongwy@chinatelecom.cn; zhaosong1@chinatelecom.cn; qibi@chinatelecom.cn)} %
	\thanks{R.-S. Han is with Cloud Network Operating System R\&D Center, China Telecom, Beijing 102209, China (E-mail: hanruisi@chinatelecom.cn)} %
	\thanks{S. Chen is with the School of Electronics and Computer Science, University of Southampton, Southampton SO17 1BJ, U.K. (E-mail: sqc@ecs.soton.ac.uk).} %
		\thanks{\scriptsize Digital Object Identifier: 10.1109/TMC.2026.3688690. 
		\copyright~2026 IEEE. Personal use of this material is permitted.
		Permission from IEEE must be obtained for all other uses, in any current or future media,
		including reprinting/republishing this material for advertising or promotional purposes,
		creating new collective works, for resale or redistribution to servers or lists,
		or reuse of any copyrighted component of this work in other works.
	}
	
\vspace*{-8mm}
}

\maketitle 
\begin{abstract}
	The emergence of high-frequency pulsating logistics unmanned aerial vehicle (UAV) swarms gives rise to ``Digital Tides'', i.e., complex traffic dynamics that challenge sustainable resource provisioning in mobile computing networks. Conventional infrastructure provisioning strategies, which typically rely on static snapshot-based analysis and localized density estimation, fail to capture the macroscopic advection of computational workloads. As a result, reactive resource activation suffers from inherent hysteresis, yielding nominal efficiency gains at the cost of mission-critical service loss at the advancing wavefront. To address this issue, we develop a fluid-based spatiotemporal framework by explicitly solving the continuity equation to characterize the macroscopic velocity field of the workload flow. Building on this framework, we propose a flux-aware asymmetric activation strategy that leverages the derived information flux vector as a kinematic precursor of demand propagation. Unlike symmetric thresholding, the proposed control logic decouples activation and deactivation dynamics. Theoretical analysis confirms the intrinsic spatial phase-lead of the flux signal and shows that the proposed strategy generates a proactive guard ring to compensate for service setup latency, including delays caused by mobile edge computing container cold-starts. We further derive closed-form expressions for instantaneous service availability and period-average energy efficiency. In addition, we formulate a quality-of-service-penalized metric to evaluate effective energy efficiency under strict outage constraints. Numerical results show that the proposed flux-driven strategy enables zero-latency tracking of the mobile wavefront and achieves a Pareto-optimal trade-off between service reliability and energy consumption, outperforming reactive baselines in dynamic logistics corridors.
\end{abstract}

	
\section{Introduction}\label{sec:intro} 

\IEEEPARstart{T}{he} rapid proliferation of logistics unmanned aerial vehicles (UAVs) for food delivery, medical transport, etc. is reshaping the low-altitude economy, establishing aerial mobile computing as a cornerstone of future autonomous systems \cite{Azari2022, Geraci2022, DongTCOM, DongJSAC}. Unlike terrestrial users with stochastic mobility, logistics UAVs exhibit highly structured macroscopic flow patterns driven by on-demand delivery cycles. This creates a phenomenon we term the low-altitude digital tide, a pulsating traffic pattern where swarms periodically disperse from distribution centers and return to base. This tidal traffic pattern involves two critical kinematic phases: the expansion phase where the spatial footprint of the swarm grows rapidly, and the contraction phase where the swarm recedes. Supporting this complex tidal traffic poses a significant energy challenge. Ground infrastructures, such as mobile edge computing (MEC)-enabled base stations (BSs), typically employ always-on service interfaces to guarantee availability. Given that infrastructure provisioning accounts for over 80\% of network energy consumption \cite{Gang2015}, maintaining idle computing resources during the tide's ebb periods fundamentally contradicts the sustainability goals of mobile computing systems.

Simply applying conventional energy-saving strategies to this scenario creates a critical reliability conflict. Existing methods, often rooted in static snapshot-based analysis \cite{IoTJDong, DongGC, DongGC2}, predominantly rely on detecting local workload density. While effective for stationary hotspots, this static perspective is inherently ill-suited for the application scenario considered. The core failure mechanism lies in the mismatch between the resource provisioning speed and the workload advection speed. Specifically, as a UAV swarm propagates outward, it forms a rapidly moving demand wavefront. However, due to the inevitable service setup latency specified by the necessary transition time from sleep to active mode, including MEC container cold-starts, a node triggered solely by local density will activate too late, often after the leading UAVs have already entered its service area. This hysteresis results in a wavefront outage, creating a moving service void that compromises the reliability of mission-critical logistics operations.

While data-driven prediction methods attempt to mitigate this lag, they often incur prohibitive computational overheads for training and lack the generalizability to handle emergent swarm behaviors. To overcome this latency-induced failure, resource orchestration must shift from observing static snapshots to predicting dynamic workload flows. This paper introduces a fluid-dynamic provisioning framework for sustainable mobile computing infrastructure. Instead of restricting the analysis to independent static realizations common in traditional stochastic modeling, we model the logistics swarm as a continuous compressible fluid governed by conservation laws \cite{sergiou2013estimating, skehill2007application, chiasserini2007fluid}. This allows us to derive an information-flux vector field that quantifies the momentum of computational demand. Moreover, unlike the scalar density, which peaks only upon arrival, the flux vector naturally points in the direction of flow and reaches its threshold ahead of the density peak. By exploiting this intrinsic kinematic phase lead, our strategy effectively generates a time buffer, waking up infrastructure resources just-in-time to reliably serve the incoming wavefront.

\subsection{Related Works}\label{sec:related_works} 

Research on energy efficiency (EE) in non-terrestrial networks (NTNs) primarily focuses on two optimization domains: aerial agents and terrestrial infrastructure. 
The majority of existing literature prioritizes UAV endurance, extensively exploring joint trajectory design and power control to minimize propulsion and operational energy consumption \cite{zhou2022uav, Zhang2021}. 
Conversely, studies focusing on terrestrial networks typically exploit UAVs as aerial edge nodes to offload workloads, thereby enabling ground BSs (GBSs) to enter sleep modes \cite{li2025uav, Liu2023}. 
However, the reciprocal challenge—optimizing the energy consumption of ground infrastructure dedicated to serving high-density aerial logistics—has received limited attention. 
Although early works such as \cite{chowdary2021joint} investigated GBS sleeping for UAV command and control, they relied on static access techniques like coordinated multi-point and non-orthogonal multiple access. 
These approaches treat aerial workloads as stationary spatial distributions, failing to capture the macroscopic advection inherent in logistics tides. 
Consequently, they lack the predictive capability necessary to compensate for service setup latency.

Adapting dynamic resource-scaling strategies from terrestrial mobile networks offers a potential solution. These strategies are generally categorized into robust optimization and learning-based prediction. Robust optimization frameworks, such as chance-constrained programming \cite{teng2021joint}, address workload uncertainty via statistical moments.  While robust against random fluctuations in stationary terrestrial clients, these methods typically adopt a discrete time-slotted perspective, which lacks the temporal resolution required to track a rapidly propagating demand wavefront.  This leads to inherent hysteresis-induced outages. Similarly, data-driven paradigms, including long short-term memory (LSTM) networks \cite{wang2024base, jang2020base, azari2022energy} and deep reinforcement learning (DRL) \cite{wu2021deep}, depend on continuous workload observability, a condition unavailable to sleeping nodes.  Furthermore, these models often lack the physical generalizability required to handle emergent swarm behaviors on unseen logistics routes \cite{wu2021deep}.

The ineffectiveness of adapting the aforementioned terrestrial strategies to the logistics UAV swarm context is fundamentally rooted in their mathematical modeling of system workload. 
Standard stochastic geometry analysis \cite{Baccelli2009, Peng2013, Shabbir2024,Zhaoetal2024}, widely adopted to derive EE limits in mobile networks, relies heavily on the snapshot assumption, treating network evolution as a sequence of independent static realizations. 
Although advances in spatio-temporal stochastic geometry have begun to address temporal interference correlations \cite{Win2009, Haenggi2013, Lu2021}, these frameworks predominantly focus on stationary users or random waypoint mobility. 
They inherently lack the capability to capture the structured macroscopic advection characteristic of logistics tides. 
Consequently, such static modeling frameworks fail to provide the kinematic foresight needed to preemptively overcome service setup latency.

To address the limitations of discrete agent modeling, fluid dynamic models were proposed to capture macroscopic mobility patterns \cite{sergiou2013estimating, skehill2007application, chiasserini2007fluid}.  Similarly, mean field games were employed to mathematically formalize the interactions of massive UAV agents as a continuum \cite{Shiri2020, Chen2021, Wang2024}.  However, these works primarily focus on distributed trajectory optimization or power control, leaving the coupled dynamics between macroscopic workload flow and ground resource activation unexplored. Prior studies \cite{sergiou2013estimating, skehill2007application} utilized the fluid principles for capacity estimation or trajectory synthesis. These contributions, however, remain largely descriptive rather than prescriptive. A critical gap persists in mapping continuous kinematic metrics, such as flow velocity and flux, to the discrete switching dynamics of resource activation. Mitigating service setup latency \cite{Onireti2017} in high-mobility corridors remains an open engineering challenge. This limitation leads to a deceptive reliability-efficiency trade-off, where reactive strategies achieve misleadingly high nominal EE by systematically shedding the energy-intensive demand wavefront. To bridge this gap, we propose a fluid-based spatiotemporal framework that leverages information flux as a kinematic precursor. Unlike reactive strategies limited by snapshot-based inputs, our approach proactively compensates for service setup latency by aligning network activation with the macroscopic momentum of the digital tide.

\subsection{Our Contributions}\label{S1.2}

Our main contributions are summarized as follows.
\begin{itemize}
	\item We establish a fluid-dynamic traffic model governed by conservation laws to characterize the pulsating dynamics of logistics UAV swarms. By explicitly solving the continuity equation, we derive the closed-form macroscopic velocity field. This continuous formulation overcomes the static limitations of classical stochastic modeling, establishing a fluid-based spatiotemporal framework that captures the cumulative spatiotemporal workload volume for evaluating non-stationary infrastructure sustainability.
	
	\item We propose a flux-aware asymmetric activation strategy that bridges the gap between macroscopic flow dynamics and discrete resource states. Identifying the information flux vector as a kinematic precursor for demand momentum, the proposed control logic exploits the intrinsic spatial phase-lead at the wavefront. This mechanism generates a proactive guard ring that effectively compensates for service setup latency, thereby enabling zero-latency tracking of the expanding demand wavefront.
	
	\item We develop a tractable analytical framework to quantify system performance, deriving closed-form expressions for the instantaneous service availability and period-average EE. To resolve the deceptive reliability-efficiency trade-off suffered by existing strategies, we formulate a quality of service (QoS)-penalized effective EE metric. Guided by this metric, we demonstrate that the proposed strategy can be calibrated to identify the Pareto-optimal operating point, effectively decoupling the conflict between energy conservation and service reliability.
\end{itemize}
From a broader perspective, the fluid-spatiotemporal viewpoint underlying the Digital Tides framework was subsequently formalized into the fluid-spatiotemporal stochastic geometry (F-STSG) framework for modeling information flow in non-stationary fields~\cite{Dong2026FSTSG}. This continuum-field perspective was further extended to wireless routing through the vorticity dissipation based routing (VDR) framework, which exploits flow decomposition and vorticity dissipation for loop-free transport in ultra-dense networks~\cite{Dong2026VDR}.

This paper is organized as follows. Section \ref{sec:background} introduces the fundamental concepts of digital tides and motivates the adoption of the fluid-dynamic paradigm. Section \ref{sec:system_model} establishes the system model, formulates the low-altitude digital tide using fluid dynamics and introduces the hardware constraints. Section \ref{sec:system_strategy} details the flux-aware asymmetric activation strategy. Section \ref{sec:performance_analysis} provides the theoretical performance analysis, deriving the resource activation geometry, service availability, and EE, alongside the proof of the flux phase-lead property. Section \ref{sec:numerical_results} presents numerical results, validating the analytical models and quantifying the performance gains regarding outage elimination and Pareto optimality. Section \ref{sec:discussion} discusses the practical limitations and the generalizability of the proposed framework under non-ideal conditions. Finally, Section \ref{sec:conclusion} concludes the paper.

Throughout this paper, $\mathbb{P}(\cdot)$ and $\mathbb{E}[\cdot]$ denote the probability and expectation operators, respectively. $\mathbb{I}(\cdot)$ denotes the indicator function, which equals 1 if the condition inside is true, and 0 otherwise. $\mathcal{L}_{I}(s)=\mathbb{E}[e^{-sI}]$ is the Laplace transform of the interference $I$.  The probability density function (PDF) of random variable $X$ is denoted by $f_x(x)$. ${}_2F_1(a, b; c; z) = \sum_{n=0}^{\infty} \frac{(a)_n \, (b)_n}{(c)_n \, n!} \, z^n$ is the standard mathematical notation for the Gaussian hypergeometric function, and $[\cdot]^+ \triangleq \max\{0, \cdot\}$. For easy reference, Table~\ref{Table-Notation} summarizes the key symbols used in this paper. 

\begin{table}[!t]
	\vspace*{-2mm}
	\scriptsize
	\caption{Summary of Key Mathematical Notations}
	\label{Table-Notation} 
	\centering
	\renewcommand{\arraystretch}{1.2}
	\vspace*{-2mm}
	\begin{tabular}{ll} 
		\toprule
		\textbf{Symbol} & \textbf{Description} \\
		\midrule
		$\Phi_{\textrm{BS}}, \lambda_{\textrm{BS}}$ & PPP and density of GBSs \\
		$H_{\textrm{UAV}}$ & Operational altitude of logistics UAVs \\
		$P_{\textrm{act}}, P_{\textrm{slp}}$ & Active and sleep power consumption of the LAS module \\
		$\nu, \gamma$ & Path loss exponent and SINR threshold \\
		$B, \eta_{\rm SE}$ & Bandwidth and target spectral efficiency \\
		$\lambda(r, t)$ & Spatio-temporal workload density at distance $r$ and time $t$ \\
		$N(t)$ & Total workload demand \\
		$\sigma(t), \dot{\sigma}(t)$ & Spatial spread parameter and its expansion rate \\
		$N_0, \sigma_0$ & Baseline workload intensity and spatial spread \\
		$\delta_N, \delta_{\sigma}$ & Load modulation index and spatial expansion index \\
		$v_r(r,t)$ & Macroscopic radial expansion velocity \\
		$\alpha(\mathbf{x},t)$ & Binary activation state of a GBS at location $\mathbf{x}$ \\
		$\mathbf{\Phi}(\mathbf{x},t)$ & Information flux vector quantifying demand momentum \\
		$R_{\textrm{wf}}(t)$ & Logistic wavefront position \\
		$\lambda_{\textrm{act}}, \lambda_{\textrm{hold}}$ & Static activation and holding density thresholds \\
		$\delta_{\textrm{th}}$ & Flux sensitivity threshold  \\
		$\tau_{\textrm{boot}}$ & Service setup latency \\
		$R_{\textrm{active}}(t)$ & Effective activation radius of the network \\
		$\mathcal{G}_{\textrm{flux}}$ & Spatial flux gain  \\
		$\mathcal{T}_{\textrm{total}}, E_{\textrm{total}}$ & Total served workload volume and energy consumption \\
		$\eta_{\textrm{EE}}, \eta_{\textrm{eff}}$ & Raw energy efficiency and effective energy efficiency \\
		\bottomrule
	\end{tabular}
	\vspace*{-3mm}
\end{table}

\section{Background and Motivation}\label{sec:background} 

To establish the physical intuition behind the proposed methodology, this section introduces the fundamental concepts of digital tides and clarifies the motivation for adopting a fluid-dynamic perspective.

\subsection{From Stochastic Mobility to Digital Tides}\label{S2.1}

Traditional mobile networks are primarily designed to accommodate independent users with stochastic mobility patterns, such as random waypoint or Brownian motion. Under such microscopic randomness, the aggregate spatial distribution of users remains statistically stationary over short time scales. 

However, the deployment of logistics UAV swarms introduces a distinct macroscopic dynamic. Governed by centralized scheduling and delivery deadlines, these swarms exhibit structured macroscopic flow patterns rather than isolated random walks. When a logistics hub dispatches a wave of delivery UAVs, it creates a pulsating traffic surge propagating outward, which we conceptualize as a digital tide. Serving this tide requires the ground infrastructure to activate and deactivate synchronously with the swarm propagation.

\subsection{Limitations of Reactive Density Detection}\label{S2.2}

Evaluating dynamic infrastructure activation requires an accurate characterization of the spatial workload. Conventional resource provisioning strategies, often modeled via static stochastic geometry, predominantly rely on localized density detection. This approach triggers BSs by independently observing the instantaneous spatial accumulation of users.

While effective for stationary hotspots, the static density metric lacks predictive capability regarding macroscopic mobility. Specifically, a BS relying on local density detection initiates its activation sequence only after the swarm enters its coverage area. Given the service setup latency associated with hardware wake-up and MEC container cold-starts, this reactive detection creates a temporal lag. Consequently, the leading edge of the swarm, namely, the mission-critical wavefront, experiences significant service outages. 

\subsection{The Fluid-Dynamic Paradigm and Information Flux}\label{S2.3}

To overcome the hysteresis in reactive control, the infrastructure must anticipate the arrival of the demand. This requires a shift in mathematical abstraction from discrete static agents to a continuous dynamic flow. By modeling the dense UAV swarm as a compressible fluid governed by the continuity equation, we analytically define the macroscopic velocity field of the traffic.

Building upon this fluid-dynamic foundation, we introduce the concept of information flux. Defined as the product of the scalar spatial density and the macroscopic velocity vector, the flux encapsulates the momentum of the service demand. The physical motivation for utilizing this metric lies in its spatial phase-lead property. At the expanding wavefront, the macroscopic velocity increases the flux signal, allowing it to reach detection thresholds at greater distances compared to the static density. By exploiting this kinematic property, the network generates a proactive guard ring, reserving the required time buffer to execute MEC cold-starts before the physical workload arrives.

\section{System Model and Problem Formulation}\label{sec:system_model} 

Fig. \ref{fig:system_model} illustrates an MEC-enabled urban logistics network driven by a pulsating logistics tide, which is conceptualized as a macroscopic spatiotemporal field of traffic density that undergoes periodic, radial expansion and contraction.

\subsection{Network Architecture}\label{sec:system_model_A} 

\subsubsection{Logistics Hub and UAV Kinematics}
The system serves a centralized logistics area $\mathcal{A} \subset \mathbb{R}^2$, targeting scenarios such as food delivery or medical transport. A logistics hub is located at the origin of $\mathcal{A}$, serving as the source and sink of the UAV swarm. UAVs operate at a constant altitude $H_{\textrm{UAV}}$. Unlike stochastic mobility models, such as random way-point, the aggregate swarm kinematics exhibit a deterministic macroscopic structure: a radial velocity field driven by the periodic waves of delivery demand.

\subsubsection{Ground Infrastructure and Hardware Constraints}
The ground network comprises cellular BSs deployed according to a homogeneous Poisson point process (PPP) $\Phi_{\textrm{BS}}$ with density $\lambda_{\textrm{BS}}$. To balance ubiquitous coverage reliability with high-capacity throughput requirements, we adopt a dual-layer heterogeneous architecture consisting of: 1) the terrestrial sector layer, where existing down-tilted sectors remain always-active to provide basic command and control links via antenna sidelobes; and 2) the low-altitude sector (LAS) layer, where BSs are equipped with dedicated up-tilted modules utilizing orthogonal frequency resources for high-bandwidth computational task offloading.

In this framework, we adopt an MEC architecture where each GBS is equipped with an edge server to handle latency-sensitive logistics tasks, such as visual navigation and path planning. Given that the radio unit and the edge computing module share a power supply and control interface, their activation is synchronized. Accordingly, the service setup latency $\tau_{\textrm{boot}}$ is modeled as the aggregate of the radio frequency circuit wake-up time and the MEC container cold-start delay, which includes the time needed to instantiate virtual machines and load service images.

This work focuses on the energy-efficient activation control of the power-consuming LAS layer. Consistent with analytical models for infrastructure sleep depth \cite{Onireti2017}, we incorporate the activation latency $\tau_{\textrm{boot}}$ to account for the non-instantaneous transition from sleep to active mode. This imposes a timing  constraint on the network control: to guarantee service availability at time $t$, the activation decision $\alpha(\mathbf{x}, t)$ must be asserted at or before $t-\tau_{\textrm{boot}}$. The specific design of the flux-aware control logic satisfying this constraint will be detailed in Section \ref{sec:system_strategy}.

\begin{figure}[!t]
	\centering
	\vspace*{-2mm}
	\includegraphics[width=\columnwidth]{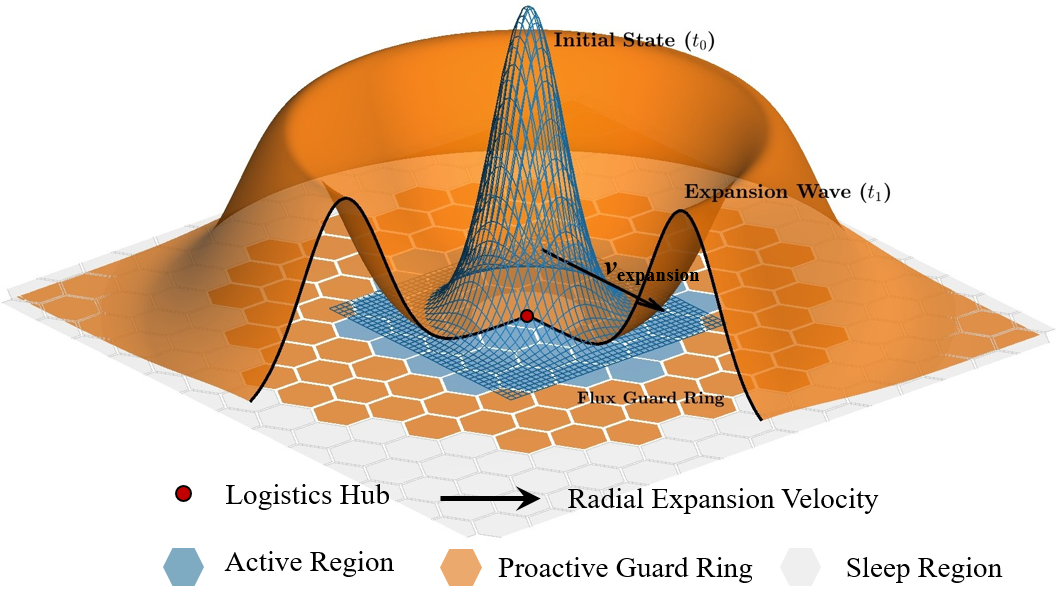} 
	\vspace*{-8mm}
	\caption{System model of the MEC-enabled UAV logistics network.}
	\vspace*{-4mm}
	\label{fig:system_model} 
\end{figure}

\vspace*{-2mm}
\subsection{Fluid-Dynamic Workload Modeling}\label{sec:system_tide} 

To capture the macroscopic mobility of the UAV swarm, we model the traffic as a compressible fluid driven by logistics demand.

\subsubsection{Density Field Formulation}

We model the instantaneous mobile agent density $\lambda(r,t)$ at distance $r$ and time $t$ as a time-varying, non-separable Gaussian field:
\begin{equation} 
	\lambda(r, t) = \frac{N(t)}{2\pi \sigma^2(t)} \exp\left( - \frac{r^2}{2\sigma^2(t)} \right).
	\label{eq:gaussian_intensity}
\end{equation}
To explicitly capture the periodicity, we model the load $N(t)$ and spread $\sigma(t)$ as sinusoidal functions with period $T_{\textrm{period}}$:
\begin{align} 
	N(t) &= N_0 \left( 1 + \delta_N \cos(\omega t) \right), \label{eq:Nt_def} \\
	\sigma(t) &= \sigma_0 \left( 1 + \delta_\sigma \cos(\omega t + \varphi) \right), \label{eq:sigmat_def}
\end{align}
where $N_0$ and $\sigma_0$ are the baseline workload intensity and spatial spread, respectively, $\delta_N$ and $\delta_\sigma$ are the workload modulation index and spatial expansion index, while $\omega\! =\! 2\pi/T_{\textrm{period}}$ is the angular frequency, and $\varphi$ accounts for the phase synchronization offset between traffic generation and spatial expansion.

\subsubsection{Macroscopic Velocity and Information Flux}

Governed by the law of conservation of mass, the evolution of the density field $\lambda(\mathbf{x}, t)$ at position $\mathbf{x}\! \in\! \mathbb{R}^2$ and time $t$ adheres to the continuity equation:
\begin{equation}\label{eq:continuity} 
	\frac{\partial \lambda}{\partial t} + \nabla \cdot (\lambda \mathbf{v}) = 0,
\end{equation}
where $\mathbf{v}(\mathbf{x}, t)$ is the macroscopic velocity field. By solving this equation (see Section S-I of the Supplementary File), we derive the radial expansion velocity of the logistics tide as $v_r(r,t) \approx r \frac{\dot{\sigma}(t)}{\sigma(t)}$, where $r\! =\! \|\mathbf{x}\|$ is the radial distance from the logistics hub.

\begin{remark}
The derivation in Section S-I reveals that the macroscopic flow is governed by two orthogonal forces: 1) source dynamics driven by load variation $\dot{N}$, representing a ``fictitious velocity'' required to accommodate node creation or removal; and 2) transport dynamics driven by spatial reshaping $\dot{\sigma}$. In high-mobility logistics scenarios, the kinematic transport dominates the source variation, i.e., $|\dot{\sigma}/\sigma|\! \gg\! |\dot{N}/N|$. This dominance ensures that the flow physically manifests as a uniform expansion rather than a localized fluctuation, thereby validating the use of the linear velocity model $v_r \propto r$ for proactive control.
\end{remark}

Based on this kinematic description, we introduce the information flux to quantify the momentum of service demand.

\begin{Definition}
\label{def:flux} 
The information flux, denoted by $\mathbf{\Phi}(\mathbf{x}, t)$, is a vector field representing the rate of traffic demand flow passing through a unit spatial cross-section. Mathematically, it is defined as the product of the scalar spatio-temporal user density $\lambda(\mathbf{x},t)$ and the macroscopic velocity vector $\mathbf{v}(\mathbf{x},t)$:
\begin{equation} 
	\mathbf{\Phi}(\mathbf{x},t) \triangleq \lambda(\mathbf{x},t) \mathbf{v}(\mathbf{x},t) \approx \lambda(r,t) \left( r \frac{\dot{\sigma}(t)}{\sigma(t)} \right) \mathbf{e}_r ,
	\label{eq:flux_def}
\end{equation}
where $\mathbf{e}_r\! =\! \mathbf{x}/\|\mathbf{x}\|$ is the unit radial vector pointing outwards from the origin.
\end{Definition}

This vector formulation offers a critical physical insight that cannot be obtained from scalar density alone. While $\lambda(\mathbf{x},t)$ merely measures the static accumulation of load, $\mathbf{\Phi}(\mathbf{x},t)$ quantifies the dynamic trend of the demand momentum. The magnitude $\|\mathbf{\Phi}\|$ is amplified by the expansion velocity $v_r$, which scales linearly with distance $r$. At the expanding wavefront where $r$ is large, the flux intensity reaches its detection threshold before the workload density accumulates to the activation level. This kinematic property creates a spatial phase lead, serving as a predictive indicator for proactive activation.

\subsubsection{Logistic Wavefront}

We define the logistic wavefront, denoted as $R_{\textrm{wf}}(t)$, as the spatial boundary where the user density decays to the minimum service threshold $\lambda_{\textrm{th}}$:
\begin{equation} 
	R_{\textrm{wf}}(t) \triangleq \left\{ r \in \mathbb{R}^+ : \lambda(r,t) = \lambda_{\textrm{th}} \right\}.
	\label{eq:wavefront_definition}
\end{equation}
Physically, $R_{\textrm{wf}}(t)$ delineates the effective edge of the serving swarm. To quantify its propagation dynamics, we define the wavefront phase velocity as the time derivative of the wavefront position:
\begin{equation} 
	v_{\textrm{wf}}(t) \triangleq \frac{\mathrm{d} R_{\textrm{wf}}(t)}{\mathrm{d} t}.
	\label{eq:vwf_def}
\end{equation}
During the expansion phase where $\dot{\sigma} > 0$, this wavefront propagates outward with a positive velocity $v_{\textrm{wf}}(t) > 0$. A critical implication is that any reactive strategy relying solely on the local density condition $\lambda \ge \lambda_{\textrm{th}}$ will track this moving boundary with a temporal lag. This necessitates a proactive mechanism to compensate for the activation latency.

\vspace*{-2mm}
\subsection{Adiabatic Link Quality Model and SINR}\label{S3.3}

To bridge the continuous macroscopic mobility and discrete microscopic transmission, we invoke the adiabatic assumption, also known as the time-scale separation principle. Analogous to the quasi-static block-fading model \cite{Tse2005}, this assumption holds because the task offloading time $T_{\textrm{pkt}}$ is orders of magnitude shorter than the characteristic evolution scale of the logistics tide $T_{\textrm{hydro}}$. Physically, $T_{\textrm{hydro}}$ characterizes the time required for significant macroscopic reshaping, defined as the inverse of the maximum relative expansion rate: $T_{\textrm{hydro}}\! \sim\! (\max |\dot{\sigma}/\sigma|)^{-1}\! \propto\! T_{\textrm{period}}$. Since the delivery cycle $T_{\textrm{period}}$ operates on hourly scale while $T_{\textrm{pkt}}$ is in milliseconds, $T_{\textrm{pkt}}\! \ll\! T_{\textrm{hydro}}$ is strictly satisfied. Consequently, the spatial intensity field $\lambda(\mathbf{x},t)$ and the aggregate interference are treated as time-invariant snapshots during each transmission interval.

Furthermore, the set of potential interferers is determined by the instantaneous binary activation states $\alpha(\mathbf{X}_i, t) \in \{0, 1\}$:
\begin{equation} 
	\Phi_{\textrm{active}}(t) = \{ \mathbf{X}_i \in \Phi_{\textrm{BS}} \mid \alpha(\mathbf{X}_i, t) = 1 \},
\end{equation}
where $\mathbf{X}_i \in \mathbb{R}^2$ denotes the spatial location of the $i$-th GBS.
In contrast to static homogeneous PPP networks, $\Phi_{\textrm{active}}(t)$ constitutes a time-varying, spatially inhomogeneous point process governed by the pulsating logistics tide.

For a typical UAV located at $\mathbf{x}_0$, the instantaneous signal-to-interference-plus-noise ratio (SINR) is given by:
\begin{equation}\label{eqSINR} 
	\textrm{SINR}(\mathbf{x}_0, t) = \frac{P_{\textrm{LAS}} h_0 L(\|\mathbf{x}_0 - \mathbf{X}^*\|)}{I_{\textrm{agg}}(\mathbf{x}_0, t) + P_{\mathrm{N}}},
\end{equation}
where $\mathbf{X}^* \in \Phi_{\textrm{active}}(t)$ is the serving BS that provides the maximum received power, and $P_{\mathrm{N}}$ is the noise power, while the aggregate interference $I_{\textrm{agg}}(\mathbf{x}_0, t)$ is the summation of signals from all other active BSs within the active region (AR) $\mathcal{A}_{\mathrm{on}}(t)$:
\begin{equation} 
	I_{\textrm{agg}}(\mathbf{x}_0, t) = \sum_{\mathbf{X}_i \in \Phi_{\textrm{active}}(t) \setminus \{\mathbf{X}^*\}} P_{\textrm{LAS}} h_i L(\|\mathbf{x}_0 - \mathbf{X}_i\|).
	\label{eq:interference_agg}
\end{equation}
In (\ref{eqSINR}) and (\ref{eq:interference_agg}), $h_0$ and $h_i$ are the respective small-scale fading channel power gains, which are exponentially distributed for Rayleigh fading, $L(d)\! =\! d^{-\nu}$ is the power-law path loss function with exponent $\nu > 2$, and $P_{\textrm{LAS}}$ denotes the transmission power of the LAS module.

\vspace*{-2mm}
\subsection{Spatiotemporal Performance Metrics}\label{sec:system_metrics} 

To quantify the trade-off between sustainability and reliability under dynamic load, we define three key metrics over a logistics cycle $T_{\textrm{period}}$. Let $\boldsymbol{\alpha}\! =\! \{\alpha(\mathbf{x},t)\! \in\! \{0,1\} \mid \mathbf{x}\! \in\! \mathcal{A}, t\! \in\! [0, T_{\textrm{period}}]\}$ be the binary activation policy of the LAS layer.

\subsubsection{Total Served Workload Volume}

The primary utility of the logistics network is the computational tasks offloaded by the UAVs. Unlike static networks where the instantaneous sum-rate suffices, the non-stationary logistics tide requires evaluating the cumulative service volume. We define the total served traffic, denoted as $\mathcal{T}_{\textrm{total}}(\boldsymbol{\alpha})$, as the spatio-temporal integration of the successful service capacity density:
\begin{equation} 
	\mathcal{T}_{\textrm{total}}(\boldsymbol{\alpha}) \triangleq \int_0^{T_{\textrm{period}}} \int_{\mathcal{A}} \mathcal{R}(\mathbf{x}, t; \boldsymbol{\alpha}) \, \mathrm{d}\mathbf{x} \, \mathrm{d}t,
	\label{eq:T_total_def}
\end{equation}
where $\mathcal{R}(\mathbf{x}, t; \bm{\alpha})$ represents the successful service capacity density, also known as the effective spatial throughput, formulated as $\mathcal{R}(\mathbf{x}, t; \bm{\alpha}) = B\, \eta_{\mathrm{SE}} \, \lambda(\mathbf{x},t) \mathsf{P}_{\mathrm{cov}}(\mathbf{x}, t \mid \boldsymbol{\alpha})$. Here, $\mathsf{P}_{\mathrm{cov}}(\mathbf{x}, t \mid \boldsymbol{\alpha}) = \mathbb{P}(\mathrm{SINR}(\mathbf{x},t) > \gamma)$ denotes the service availability which is mathematically quantified by the instantaneous coverage probability, $B$ is the system bandwidth, and $\eta_{\mathrm{SE}}$ represents the target spectral efficiency.

In the context of mission-critical logistics, this link-level coverage probability serves as the fundamental metric quantifying the system-level service availability.

\subsubsection{Total Energy Consumption}

The network energy consumption is directly determined by the spatiotemporal activation pattern of the GBSs. The total energy consumption over one period is defined as:
\begin{equation} 
	E_{\textrm{total}}(\boldsymbol{\alpha}) \triangleq \int_0^{T_{\textrm{period}}} \int_{\mathcal{A}} \lambda_{\textrm{BS}} P_{\textrm{GBS}}(\alpha(\mathbf{x},t)) \, \mathrm{d}\mathbf{x} \, \mathrm{d}t,
	\label{eq:E_total_def}
\end{equation}
where $P_{\textrm{GBS}}(\cdot)$ models the power consumption of a single LAS module:
\begin{equation} 
	P_{\textrm{GBS}}(\alpha) = P_{\textrm{act}} \cdot \alpha + P_{\textrm{slp}} \cdot (1 - \alpha).
\end{equation}
Here, $P_{\textrm{act}}$ and $P_{\textrm{slp}}$ denote the power consumption in the active and sleep modes, respectively.

\subsubsection{Service Unavailability (Wavefront Outage)}

To quantify the reliability risk at the expanding edge, we define the instantaneous service unavailability, $\mathcal{O}(t)$, as the fraction of workload mass located in the ``outage zone''—the spatial gap between the physical demand wavefront $R_{\textrm{wf}}(t)$ and the effective service boundary $R_{\textrm{active}}(t)$ induced by the control policy $\boldsymbol{\alpha}$.

Mathematically, assuming that the policy generates a contiguous coverage area of radius $R_{\textrm{active}}(t)$, $\mathcal{O}(t)$ is calculated as the normalized user mass within the gap $[R_{\textrm{active}}(t), R_{\textrm{wf}}(t)]$:
\begin{equation} 
	\mathcal{O}(t; \boldsymbol{\alpha}) \triangleq \frac{\left[ \int_{R_{\textrm{active}}(t)}^{R_{\textrm{wf}}(t)} \lambda(r,t) \cdot 2\pi r \, \mathrm{d}r \right]^+}{\int_{0}^{R_{\textrm{wf}}(t)} \lambda(r,t) \cdot 2\pi r \, \mathrm{d}r} .
	\label{eq:outage_metric}
\end{equation}
The unavailability is zero if the service boundary exceeds the wavefront, i.e., $R_{\textrm{active}} \ge R_{\textrm{wf}}$. This metric specifically captures the percentage of the logistics swarm—most notably the strategic leading edge—that is physically present but disconnected due to infrastructure hysteresis.

\vspace*{-2mm}
\subsection{QoS-Constrained EE Maximization Problem}\label{sec:problem_formulation} 

\subsubsection{Objective Function and Constraints}

The primary objective of sustainable mobile computing networks is to maximize the network EE, $\eta_{\textrm{EE}}(\boldsymbol{\alpha})$, which is the ratio of the total served workload volume to the total energy consumption over the logistics cycle:
\begin{equation} 
	\eta_{\textrm{EE}}(\boldsymbol{\alpha}) \triangleq \frac{\mathcal{T}_{\textrm{total}}(\boldsymbol{\alpha})}{E_{\textrm{total}}(\boldsymbol{\alpha})}.
\end{equation}
Mathematically, we formulate the network control problem as a QoS-constrained EE maximization problem:
\begin{subequations} 
	\label{eq:opt_problem_all}
	\begin{align}
		\mathcal{P}_0: \quad \max_{\boldsymbol{\alpha}} \quad & \eta_{\textrm{EE}}(\boldsymbol{\alpha}), \label{eq:obj_func} \\
		\text{s.t.} \quad & \mathcal{O}(t; \boldsymbol{\alpha}) \le \epsilon_{\textrm{th}}, \quad \forall t \in \mathbb{T}_{\text{exp}}, \label{eq:constraint_qos} \\
		& \alpha(\mathbf{x},t) \in \{0, 1\}, \quad \forall \mathbf{x}, t, \label{eq:constraint_binary}
	\end{align}
\end{subequations}
where $\mathbb{T}_{\text{exp}}$ denotes the time interval corresponding to the expansion phase of the logistics tide, and $\epsilon_{\textrm{th}}$ is the maximum allowed service unavailability threshold, e.g., $\epsilon_{\textrm{th}}\! =\! 0.01$.

\subsubsection{Generalized Effective Energy Efficiency}

Problem $\mathcal{P}_0$ is a non-convex fractional programming problem. Furthermore, the constraint \eqref{eq:constraint_qos} imposes an infinite number of instantaneous QoS requirements over the continuous domain $\mathbb{T}_{\text{exp}}$, rendering direct global optimization intractable. 

To obtain a tractable formulation, we adopt a penalty-function method to transform the constrained optimization problem into an unconstrained maximization of a surrogate objective function, termed the generalized effective EE, denoted as $\eta_{\textrm{eff}}$.

\begin{Definition}
\label{def:effective_ee} 
We define the effective EE as the following QoS-penalized metric:
\begin{equation} 
	\eta_{\mathrm{eff}} \triangleq \eta_{\mathrm{EE}} \cdot \mathcal{U}_{\beta}\left( \sup_{t \in \mathcal{T}_{\mathrm{exp}}} \mathcal{O}(t) \right).
	\label{eq:effective_ee}
\end{equation}
Here, $\mathcal{U}_{\beta}(\cdot)$ serves as a soft barrier function, or a multiplicative penalty function, derived from the reliability constraint. We adopt the polynomial form:
\begin{equation} 
	\mathcal{U}_{\beta}(x) = 
	\begin{cases}
		(1 - x)^\beta, & \text{if } 0 \le x < 1, \\
		0, & \text{if } x \ge 1,
	\end{cases}
\end{equation}
where $\beta \ge 1$ is the penalty factor.
\end{Definition}

\begin{remark}
The formulation of $\eta_{\mathrm{eff}}$ reconciles the conflicting goals of $\mathcal{P}_0$. Physically, the term $(1-x)$ represents the service availability. As the penalty factor $\beta$ increases, the cost of violating the outage constraint \eqref{eq:constraint_qos} becomes severe, forcing the optimal solution of the unconstrained proxy problem to asymptotically approach the feasible region of the original problem $\mathcal{P}_0$. This conversion allows us to evaluate the system performance using a single scalar metric that  reflects the holistic operational value under QoS constraints.
\end{remark}

\vspace*{-2mm}
\section{Flux-Aware Asymmetric Activation Strategy}\label{sec:system_strategy} 

To effectively manage the high-frequency pulsating logistics tide, the network control logic must bridge the gap between macroscopic fluid dynamics and discrete hardware states. We propose a flux-aware asymmetric activation strategy. Unlike symmetric thresholding, this strategy decouples the activation and deactivation dynamics to accommodate the distinct kinematic requirements of the expansion and contraction phases.

\vspace*{-2mm}
\subsection{Principles of Asymmetric Activation}\label{S4.1}

The core idea of the proposed strategy is to align the network response with the distinct risk profiles inherent in the pulsating workload. Rather than applying a uniform control logic, we design tailored state-transition mechanisms specifically adapted to the unique vulnerabilities of the expansion and contraction phases.

\subsubsection{Expansion Phase ($\dot{\sigma} > 0$)} 

The primary risk is the wavefront outage caused by boot-up latency. To mitigate this, the strategy employs a proactive activation mechanism, leveraging the spatial phase-lead of the information flux to trigger GBSs before the workload density physically accumulates.

\subsubsection{Contraction Phase ($\dot{\sigma} \le 0$)}

The primary challenge is maintaining connectivity for agents located at the spatial periphery, i.e., the Gaussian tail, as the swarm recedes. To mitigate the risk of premature disconnection, the strategy employs a hysteresis retention mechanism. The flux trigger is deactivated because the inward-pointing macroscopic velocity vector no longer serves as a valid precursor for proactive activation. Instead, the control logic relies exclusively on a lowered density holding threshold, ensuring that the effective service boundary contracts at a slower rate than the physical decay of user density.

\vspace*{-2mm}
\subsection{Triggering Logic Formulation}\label{S4.2}

Let $\alpha(\mathbf{x}, t)\! \in\! \{0, 1\}$ denote the binary activation state of a GBS at location $\mathbf{x}$. We construct the state transition logic based on three distinct Boolean triggering conditions: the proactive flux trigger $\mathcal{C}_{\textrm{flux}}$, the reactive load trigger $\mathcal{C}_{\textrm{load}}$, and the retention trigger $\mathcal{C}_{\textrm{hold}}$, which are defined by:
\begin{equation} 
	\begin{cases}
		\mathcal{C}_{\textrm{flux}}: & \|\mathbf{\Phi}(\mathbf{x}, t)\| \ge \delta_{\textrm{th}}, \\
		\mathcal{C}_{\textrm{load}}: & \lambda(\|\mathbf{x}\|, t) \ge \lambda_{\textrm{act}}, \\
		\mathcal{C}_{\textrm{hold}}: & \lambda(\|\mathbf{x}\|, t) \ge \lambda_{\textrm{hold}},
	\end{cases}
\end{equation}
where $\delta_{\textrm{th}}$ is the flux sensitivity threshold, $\lambda_{\textrm{act}}$ is the activation density threshold, and we introduce a lower holding threshold $\lambda_{\textrm{hold}}\! <\! \lambda_{\textrm{act}}$, such as $\lambda_{\textrm{hold}}\! \to\! 0$,  to implement the hysteresis mechanism.
The activation protocol is defined as:
\begin{equation} 
	\alpha(\mathbf{x}, t) = 
	\begin{cases}
		\mathbb{I}\left( \mathcal{C}_{\textrm{load}} \lor \mathcal{C}_{\textrm{flux}} \right), & \text{if } \dot{\sigma}(t) > 0, \\
		\mathbb{I}\left( \mathcal{C}_{\textrm{hold}} \right) \cdot \alpha(\mathbf{x}, t-1), & \text{if } \dot{\sigma}(t) \le 0.
	\end{cases}
	\label{eq:activation_rule}
\end{equation}

To facilitate practical deployment in resource-constrained network controllers, we formalize the proposed logic in Algorithm \ref{alg:activation}.

\begin{algorithm}[!t]
	\small
	\caption{Flux-Aware Asymmetric Cycle Control}
	\label{alg:activation}
	\begin{algorithmic}[1]
		\Require Time $t$; Fluid state $\{\sigma, \dot{\sigma}, N\}$; Previous state $\boldsymbol{\alpha}_{t-1}$.
		\Ensure New activation state $\boldsymbol{\alpha}_t$.
		\State \textbf{Determine Phase:} 
		\If {$\dot{\sigma}(t) > 0$} \textit{Mode} $\leftarrow$ \textsc{Expansion}
		\Else \ \textit{Mode} $\leftarrow$ \textsc{Contraction}
		\EndIf
		\State \textbf{Update GBS States:}
		\For {each GBS $i$ at location $\mathbf{x}_i$}
		\State Compute the local density $\lambda(\mathbf{x}_i, t)$ using \eqref{eq:gaussian_intensity}.
		\If {\textit{Mode} is \textsc{Expansion}}
		\State \Comment{Compute Flux only when needed for proactive trigger}
		\State $\|\mathbf{\Phi}(\mathbf{x}_i)\| \leftarrow \lambda(\mathbf{x}_i) \cdot \|\mathbf{x}_i\| \cdot \frac{\dot{\sigma}}{\sigma}$ 
		\State \textbf{Trigger:} $\mathcal{C}_{\textrm{flux}} \leftarrow (\|\mathbf{\Phi}(\mathbf{x}_i)\| \ge \delta_{\textrm{th}})$
		\State $\alpha(\mathbf{x}_i, t) \leftarrow \mathbb{I}(\lambda(\mathbf{x}_i) \ge \lambda_{\textrm{act}} \lor \mathcal{C}_{\textrm{flux}})$
		\Else \ \Comment{Hysteresis Retention Phase}
		\State $\alpha(\mathbf{x}_i, t) \leftarrow \mathbb{I}(\lambda(\mathbf{x}_i) \ge \lambda_{\textrm{hold}}) \cdot \alpha(\mathbf{x}_i, t-1)$
		\EndIf
		\EndFor
		\State \Return $\boldsymbol{\alpha}_t$
	\end{algorithmic}
\end{algorithm}

\vspace*{-2mm}
\subsection{Complexity and Implementation Analysis} \label{S4.3}

The proposed strategy exhibits linear computational complexity $O(M)$, where $M$ is the number of GBSs. Unlike combinatorial optimization or reinforcement learning approaches that typically scale with $O(M^2)$ or higher, our method relies on calculating analytical fluid metrics ($\lambda, \|\mathbf{\Phi}\|$) locally for each node. 
Beyond computational efficiency, the framework ensures minimal signaling overhead. Unlike centralized optimization schemes requiring real-time channel state information (CSI) from all UAVs, our framework operates on macroscopic statistical parameters, such as $\sigma(t), N(t)$, that evolve slowly on the scale of minutes (the logistics cycle) rather than milliseconds. 
Consequently, the broadcast of these fluid parameters entails negligible bandwidth cost, and the localized flux calculation supports a fully distributed, event-driven activation logic that requires no backhaul coordination between neighboring GBSs. This low overhead is critical for real-time resource provisioning in large-scale mobile systems.

\vspace*{-2mm}
\subsection{Spatiotemporal Mechanism Analysis}\label{S4.4}

The proposed logic results in a dynamic network topology that physically adapts to the logistics tide.

\subsubsection{Traveling Wave Activation (Expansion)}

During the expansion phase, the condition $\|\mathbf{\Phi}\| \ge \delta_{\textrm{th}}$ creates a proactive guard ring ahead of the density wavefront. This does not imply simultaneous global activation. Since the flux magnitude $\|\mathbf{\Phi}(r,t)\|$ is spatially graded (peaking at the wavefront), the activation condition is satisfied sequentially from the center to the periphery. This results in a ``traveling active zone'' that propagates ahead of the UAV swarm, ensuring resources are provisioned selectively and just-in-time.

\subsubsection{Hysteresis Retention (Contraction)}

As the swarm recedes, the flux intensity diminishes. The protocol automatically switches to the retention mode governed by $\mathcal{C}_{\textrm{hold}}$. This ensures that the active region shrinks slower than the density decay, maintaining connectivity for tail users until the local load becomes negligible.

\begin{remark}
A potential concern with kinematic triggers is the vulnerability to false alarms caused by high-velocity but low-density transient outliers. However, under the proposed fluid-dynamic framework, the activation is governed by the information flux ($\mathbf{\Phi} = \lambda \mathbf{v}$), which acts as a robust spatiotemporal momentum filter. Mathematically, the exponential spatial decay of the density field strictly dominates the linear growth of the macroscopic velocity, ensuring that the flux magnitude naturally vanishes in sparse far-field regions. Physically, constrained by the maximum aerodynamic speed of UAVs, a sparse outlier with negligible density cannot generate sufficient momentum to cross the macroscopic threshold $\delta_{\mathrm{th}}$. Consequently, such transient anomalies fail to trigger $\mathcal{C}_{\mathrm{flux}}$, thereby preventing premature infrastructure activation and energy waste.
\end{remark}

\vspace*{-2mm}
\section{Performance Analysis}\label{sec:performance_analysis} 

Given the spatial inhomogeneity of the pulsating tide, for tractability, we employ a local homogeneity approximation to derive the CP and EE metrics.

\vspace*{-2mm}
\subsection{Spatio-Temporal Activation Geometry}\label{S5.1}

The proposed asymmetric strategy yields a phase-dependent network topology. First, we formally derive the effective activation radius, denoted as $R_{\textrm{active}}(t)$.

\begin{Proposition}
\label{prop:activation_radius} 
The effective activation radius $R_{\mathrm{active}}(t)$ that adapts to the expansion-contraction cycle of the logistics tide is expressed as
\begin{equation} 
	R_{\mathrm{active}}(t) = 
	\begin{cases}
		\max \left\{ R_{\mathrm{act}}(t), R_{\mathrm{flux}}(t) \right\}, & \text{if } \dot{\sigma}(t) > 0, \\
		R_{\mathrm{hold}}(t), & \text{if } \dot{\sigma}(t) \le 0,
	\end{cases}
	\label{eq:R_active_cases}
\end{equation}
where the component radii are given by:
\begin{subequations} 
	\begin{align}
		& R_{\mathrm{act}}(t) = \sigma(t) \sqrt{2 \left[ \ln \left( \frac{N(t)}{2\pi \sigma^2(t) \lambda_{\mathrm{act}}} \right) \right]^+ }, \label{eq:R_act_def} \\
		& R_{\mathrm{hold}}(t) = \sigma(t) \sqrt{2 \left[ \ln \left( \frac{N(t)}{2\pi \sigma^2(t) \lambda_{\mathrm{hold}}} \right) \right]^+ }, \label{eq:R_hold_def} \\
		& R_{\mathrm{flux}}(t) = \sup \left\{ r \in \mathbb{R}^+ : \frac{r}{\sigma^3(t)} e^{-\frac{r^2}{2\sigma^2(t)}} = \mathcal{K}_{\mathrm{th}}(t) \right\}, \label{eq:R_flux_def}
	\end{align}
\end{subequations}
with $\mathcal{K}_{\mathrm{th}}(t)=\frac{2\pi \delta_{\mathrm{th}}}{N(t) |\dot{\sigma}(t)|}$ representing the normalized flux threshold. 
\end{Proposition}

\begin{proof}
	The detailed derivation is provided in Section S-II of the Supplementary File.
\end{proof}

This derived activation radius reveals distinct physical behaviors across the expansion-contraction cycle. 

\subsubsection{Expansion (Proactive Guard Ring)} 

During the expansion phase, i.e., $\dot{\sigma}\! >\! 0$, the macroscopic velocity $v_r\! \propto\! r$ amplifies the flux intensity at the periphery. This kinematic gain pushes the flux boundary $R_{\textrm{flux}}(t)$ beyond the static density boundary $R_{\textrm{act}}(t)$, creating a proactive guard ring of width $\Delta R_{\textrm{lead}}\! =\! R_{\textrm{flux}}\! -\! R_{\textrm{act}}$ that spatially compensates for the boot-up latency.

\subsubsection{Contraction (Hysteresis Retention)} 

During the contraction phase, i.e., $\dot{\sigma} \le 0$, the flux trigger is disabled. The activation boundary is governed strictly by the holding threshold $\lambda_{\textrm{hold}}$. Since $\lambda_{\textrm{hold}} < \lambda_{\textrm{act}}$, the condition $R_{\textrm{hold}}(t) > R_{\textrm{act}}(t)$ holds. This creates a retention zone of width $\Delta R_{\textrm{lag}} = R_{\textrm{hold}} - R_{\textrm{act}}$, preventing the premature disconnection of users at the spatial periphery as the swarm recedes.

\vspace*{-2mm}
\subsection{Local Service Availability Analysis}\label{sec:local_coverage} 

The dynamic activation of LAS modules creates a spatially finite and inhomogeneous interference field. Unlike standard stochastic geometry approaches that assume an infinite field of interferers, we explicitly model the truncation of interference at the activation boundary $R_{\textrm{active}}(t)$ to derive an accurate CP.

\begin{Theorem}
\label{thm:inst_local_cov} 
Consider the finite active region bounded by $R_{\mathrm{active}}(t)$. The instantaneous service availability $\mathrm{CP}$ for a typical UAV within the active zone is expressed as:
\begin{align} 
	\mathsf{P}_{\mathrm{cov}}(r,t) \approx \int_{0}^{R_{\mathrm{active}}(t)} & 2\pi\lambda_{\mathrm{BS}}\, y \, \exp\bigg( -\pi\lambda_{\mathrm{BS}}\, y^2 \Big( 1 \notag \\
	& + \mathcal{F}_{\mathrm{finite}}(y, R_{\mathrm{active}}) \Big) \bigg) \, \mathrm{d}y,
	\label{eq:Pcov_finite}
\end{align}
where $\mathcal{F}_{\mathrm{finite}}(y, R_{\mathrm{active}})$ is the finite interference factor explicitly capturing the boundary effect:
\begin{align} 
	&\mathcal{F}_{\mathrm{finite}}(y, R_{\mathrm{active}}) = {}_2F_1\left(1, \kappa; 1+\kappa; -\gamma \right) \notag \\
	&\quad - \left(\frac{R_{\mathrm{active}}}{y}\right)^2 {}_2F_1\left(1, \kappa; 1+\kappa; -\gamma \left(\frac{y}{R_{\mathrm{active}}}\right)^\nu \right),
	\label{eq:F_finite}
\end{align}
with $\kappa = 2/\nu$ and $\gamma$ being the SINR threshold.
\end{Theorem}

\begin{proof}
	The detailed derivation is provided in Section S-III of the Supplementary File.
\end{proof}

\begin{remark}
The second term in \eqref{eq:F_finite} represents the boundary correction. As $R_{\mathrm{active}} \to \infty$, this term vanishes, reducing to the classical infinite plane solution. However, near the network edge where $y \to R_{\textrm{active}}$, this term significantly reduces the estimated interference.
\end{remark} 

\vspace*{-2mm}
\subsection{Network Energy Efficiency}\label{S5.3}

By substituting the phase-dependent activation radius derived in Proposition \ref{prop:activation_radius} into the metric definitions, we obtain the closed-form expression for EE.

\begin{Theorem}
\label{theorem:period_ee} 
The period-average EE of the network, defined as the ratio of total served traffic to total energy consumption over a cycle $T_{\mathrm{period}}$, is given by:
\begin{equation} 
	\eta_{\mathrm{EE}} = \frac{\mathcal{T}_{\mathrm{total}}}{E_{\mathrm{total}}} = \frac{B\, \eta_{\rm SE} \int_0^{T_{\mathrm{period}}} \mathcal{M}(t) \bar{\mathsf{P}}_{\mathrm{cov}}(t) \, \mathrm{d}t} {\int_0^{T_{\mathrm{period}}} P(t) \, \mathrm{d}t},
	\label{eq:theorem_ee}
\end{equation}
where $\bar{\mathsf{P}}_{\mathrm{cov}}(t)$ denotes the spatially averaged CP over the dynamic active region $R_{\mathrm{active}}(t)$, approximated as:
\begin{equation} 
	\bar{\mathsf{P}}_{\mathrm{cov}}(t) \approx \frac{1}{1 + \mathcal{F}_{\mathrm{finite}}(R_{\mathrm{active}}(t))},
\end{equation}
with $\mathcal{F}_{\mathrm{finite}}$ given in \eqref{eq:F_finite}, while $\mathcal{M}(t)$ and $P(t)$ are the served workload mass and network power consumption defined below.

\textit{1) Served User Mass $\mathcal{M}(t)$:} This term represents the aggregate workload demand covered by the active network, derived from the Gaussian integration:
\begin{equation}	\label{eq:mt} 
	\mathcal{M}(t) = N(t) \left( 1 - \exp\left( - \frac{R_{\mathrm{active}}^2(t)}{2\sigma^2(t)} \right) \right).
\end{equation}
	
\textit{2) Network Power Consumption $P(t)$:} This term captures the instantaneous power cost of the infrastructure:
\begin{equation}\label{eq:pt} 
	P(t) = \pi R_{\mathrm{active}}^2(t) \lambda_{\mathrm{BS}} \Delta P + P_{\mathrm{base}},
\end{equation}
where $\Delta P = P_{\mathrm{act}} - P_{\mathrm{slp}}$ represents the activation power penalty, and $P_{\mathrm{base}} = |\mathcal{A}| \lambda_{\mathrm{BS}} P_{\mathrm{slp}}$ denotes the baseline sleep power of the entire service area $\mathcal{A}$.
\end{Theorem}

\begin{proof}
	The detailed derivation is provided in Section S-IV of the Supplementary File.
\end{proof}

This analytical form facilitates the numerical optimization of the flux threshold $\delta_{\mathrm{th}}$. The objective is to balance the size of the proactive guard ring, which dictates coverage reliability, against the incremental energy cost, thereby locating the Pareto-optimal operating point.

\vspace*{-2mm}
\subsection{Theoretical Analysis of Flux Phase-Lead}\label{sec:flux_gain} 

The effectiveness of the proposed strategy depends on whether the information flux $\mathbf{\Phi}$ provides an earlier detection signal than the scalar density $\lambda$ at the expanding wavefront. We formally prove this property by analyzing the spatial monotonicity of the flux-to-load ratio.

\begin{Lemma}
\label{lemma:flux_ratio} 
During the expansion phase where $\dot{\sigma}(t) > 0$, the ratio of the flux magnitude to the workload density, denoted as $\mathcal{R}_{\Phi/\lambda}(r,t)$, is a strictly monotonically increasing function of the radial distance $r$.
\end{Lemma}

\begin{proof}
Recall the definition of information flux from Definition \ref{def:flux}: $\|\mathbf{\Phi}(r,t)\| = \lambda(r,t) \cdot |v_r(r,t)|$. Substituting the derived macroscopic velocity $v_r(r,t) \approx r \frac{\dot{\sigma}}{\sigma}$, the ratio is given by:
\begin{equation} 
	\mathcal{R}_{\Phi/\lambda}(r,t) \triangleq \frac{\|\mathbf{\Phi}(r,t)\|}{\lambda(r,t)} = |v_r(r,t)| = r \frac{\dot{\sigma}(t)}{\sigma(t)}.
\end{equation}
For any fixed time instant $t$ in the expansion phase where $\dot{\sigma} > 0$, the term $\frac{\dot{\sigma}}{\sigma}$ is a positive constant with respect to space. Taking the partial derivative with respect to $r$:
\begin{equation} 
	\frac{\partial}{\partial r} \mathcal{R}_{\Phi/\lambda}(r,t) = \frac{\dot{\sigma}(t)}{\sigma(t)} > 0.
\end{equation}
Thus, the relative strength of the flux signal compared to the density signal increases linearly with distance. This completes the proof.
\end{proof}

\begin{remark}
Lemma \ref{lemma:flux_ratio} provides the theoretical basis for the kinematic phase lead. At the swarm center where $r \to 0$, density dominates. However, at the far-field wavefront where $r \gg \sigma$, the high expansion velocity amplifies the flux signal. Consequently, for a decaying Gaussian tail, the flux threshold condition $\|\mathbf{\Phi}\| \ge \delta_{\mathrm{th}}$ is satisfied at a strictly larger radius than the density condition $\lambda \ge \lambda_{\mathrm{th}}$ (assuming normalized thresholds). This ensures that the flux-triggered proactive guard ring is generated ahead of the traffic load.
\end{remark}

\subsubsection{Quantification of Spatial Flux Gain}

Based on the monotonicity proven above, we quantify the spatial gain afforded by the flux-aware strategy.

\begin{Definition}
During the expansion phase where $\dot{\sigma} > 0$, the flux gain $\mathcal{G}_{\mathrm{flux}}(t)$ is defined as the spatial difference between the proactive flux boundary and the reactive density boundary:
\begin{equation} 
	\mathcal{G}_{\mathrm{flux}}(t) \triangleq R_{\mathrm{flux}}(t) - R_{\mathrm{act}}(t).
\end{equation}
Since the macroscopic velocity $v_r(r) \approx r \dot{\sigma}/\sigma$ increases linearly with distance, the flux intensity $\|\mathbf{\Phi}\| = \lambda v_r$ decays slower than the density $\lambda$ at the wavefront. Consequently, for a properly calibrated $\delta_{\mathrm{th}}$, the condition $\mathcal{G}_{\mathrm{flux}}(t) > 0$ holds.
\end{Definition}

\begin{remark}
Physically, $\mathcal{G}_{\mathrm{flux}}$ represents the width of the proactive guard ring. From a control theory perspective, this metric represents a space-time conversion: the strategy trades spatial redundancy in the form of an extra active ring of width $\mathcal{G}_{\mathrm{flux}}$ for temporal margin. This spatial advance compensates for the inevitable time delay $\tau_{\mathrm{boot}}$ inherent in hardware state transitions.
\end{remark}

To ensure uninterrupted connectivity for the leading UAVs, the spatial gain must outpace the physical movement of the swarm during the wake-up process. This leads to the following fundamental condition for reliability.

\begin{Theorem}
\label{theorem:zero_latency} 
The wavefront outage is theoretically eliminated if and only if the flux gain satisfies:
\begin{equation} 
	\mathcal{G}_{\mathrm{flux}}(t) \ge \int_{t}^{t+\tau_{\mathrm{boot}}} v_{\mathrm{wf}}(\tau) \, \mathrm{d}\tau \approx v_{\mathrm{wf}}(t) \cdot \tau_{\mathrm{boot}},
	\label{eq:zero_latency_condition}
\end{equation}
where $v_{\mathrm{wf}}(t)$ is the wavefront phase velocity defined in \eqref{eq:vwf_def}. This inequality establishes the fundamental geometric constraint for the flux-aware activation boundary.
\end{Theorem}

	\begin{remark}
		It is crucial to note that satisfying the kinematic zero-latency condition in Theorem \ref{theorem:zero_latency} is a necessary but not sufficient condition for successful task offloading. The holistic system reliability is governed by a two-tier hierarchy. 1) Hardware Availability: The GBS must be active upon the UAV's arrival, which is addressed by the flux-aware guard ring $\mathcal{G}_{\mathrm{flux}}$; and 2) Link Reliability: The instantaneous SINR must exceed the target threshold $\gamma$. In our framework, this second tier is rigorously captured by the finite network coverage probability $\mathsf{P}_{\mathrm{cov}}$ derived in Theorem \ref{thm:inst_local_cov}. Consequently, the overall served traffic $\mathcal{T}_{\mathrm{total}}$ defined in \eqref{eq:T_total_def} explicitly integrates both the kinematic availability and the SINR-based coverage probability, ensuring that offloading failures due to signal degradation are inherently penalized.
	\end{remark}

The condition established in Theorem \ref{theorem:zero_latency} imposes a fundamental geometric constraint on the activation boundary. To translate this theoretical kinematic requirement into practical system design, we derive the upper bound for the flux sensitivity threshold.

\begin{Proposition}
\label{prop:threshold_calibration} 
To ensure zero-latency coverage given a service setup constraint $\tau_{\mathrm{boot}}$, the flux sensitivity threshold $\delta_{\mathrm{th}}$ must be calibrated to satisfy:
\begin{equation} 
	\delta_{\mathrm{th}} \le \inf_{t \in \mathbb{T}_{\mathrm{exp}}} \left\| \mathbf{\Phi}\left( R_{\mathrm{act}}(t) + v_{\mathrm{wf}}(t)\tau_{\mathrm{boot}}, t \right) \right\|.
	\label{eq:critical_threshold}
\end{equation}
This inequality serves as a design criterion: it quantifies the minimum sensitivity required to detect the precursor momentum at the specific look-ahead distance $d_{\mathrm{lead}} = v_{\mathrm{wf}} \tau_{\mathrm{boot}}$.
\end{Proposition}

\begin{proof}
From Theorem \ref{theorem:zero_latency}, the zero-latency requirement is $R_{\textrm{flux}} \ge R_{\textrm{act}} + v_{\textrm{wf}} \tau_{\textrm{boot}}$. Let the right-hand side be the required target radius $R_{\textrm{req}} \triangleq R_{\textrm{act}} + v_{\textrm{wf}} \tau_{\textrm{boot}}$. Recall that $R_{\textrm{flux}}$ is the outer root of $\|\mathbf{\Phi}(r)\| = \delta_{\textrm{th}}$. In the wavefront region where $r > \sigma$, the flux magnitude $\|\mathbf{\Phi}(r)\|$ is a monotonically decreasing function of $r$. Therefore, to ensure the trigger radius $R_{\textrm{flux}}$ extends beyond $R_{\textrm{req}}$, the triggering threshold $\delta_{\textrm{th}}$ must be lower than the flux intensity at $R_{\textrm{req}}$:
\begin{equation} 
	\delta_{\textrm{th}} = \|\mathbf{\Phi}(R_{\textrm{flux}})\| \le \|\mathbf{\Phi}(R_{\textrm{req}})\|.
\end{equation}
Substituting the definition of $\mathbf{\Phi}$ yields the bound in \eqref{eq:critical_threshold}. This completes the proof.
\end{proof}

\section{Numerical Results and Analysis}\label{sec:numerical_results} 

We validate the analytical framework using Monte Carlo simulations averaged over 50,000 iterations within a service area $\mathcal{A}$ of $20 \times 20$ km$^2$ governed by the pulsating workload model. 
	The default system parameters are listed in Table \ref{tab:sim_params}. To ensure the realism of the evaluation, these parameters are selected based on a combination of standardization reports and practical hardware specifications. Specifically, the communication parameters, including the UAV operational altitude ($150$ m) and path loss exponent ($\nu=2.5$), are aligned with the 3GPP technical report for low-altitude aerial vehicles (TR 36.777) \cite{3gpp_tr36777}. The hardware profiles, such as the active/sleep power consumption and the service setup latency ($\tau_{\mathrm{boot}} = 5$ min), reflect the typical specifications of MEC-enabled micro BSs. This service setup latency accounts for the aggregate delay of the hardware transitioning from a deep sleep mode \cite{Onireti2017} and the software overhead of MEC container cold-starts \cite{mec_latency_ref}. Furthermore, the macroscopic traffic variables ($N_0$ and $\sigma_0$) are configured as design assumptions to represent a city-scale logistics hub serving a metropolitan area.

\begin{table}[!b]
	\vspace*{-5mm}
	\scriptsize
	\caption{Default Simulation Parameters}
	\label{tab:sim_params} 
	\centering
	\vspace*{-3mm}
	\begin{tabular}{l|c|c}
		\toprule
		\textbf{Parameter} & \textbf{Symbol} & \textbf{Value} \\
		\midrule
		GBS Density (LAS Layer) & $\lambda_{\textrm{BS}}$ & $5 \text{ BSs/km}^2$ \\
		UAV Flight Altitude & $H_{\textrm{UAV}}$ & $150 \text{ m}$ \\
		Bandwidth per LAS & $B$ & $20 \text{ MHz}$ \\
		Path Loss Exponent & $\nu$ & $2.5$ \\
		Target SINR Threshold & $\gamma$ & $-5 \text{ dB}$ \\ 
		Target Spectral Efficiency & $\eta_{\rm SE}$ & $2 \text{ bps/Hz}$ \\
		Baseline Workload Intensity & $N_0$ & $25,000$ UAVs \\
		Baseline Spatial Spread & $\sigma_0$ & $4.5 \text{ km}$ \\
		Load Modulation Index & $\delta_N$ & $0.8$ \\
		Spatial Expansion Index & $\delta_\sigma$ & $0.7$ \\ 
		Active Power (LAS) & $P_{\textrm{act}}$ & $400 \text{ W}$ \\
		Sleep Power (LAS) & $P_{\textrm{slp}}$ & $50 \text{ W}$ \\
		Service Setup Latency & $\tau_{\textrm{boot}}$ & $300 \text{ s}$ ($5 \text{ min}$) \\
		Flux Sensitivity Threshold & $\delta_{\textrm{th}}$ & \textit{Variable} \\
		Static Activation Threshold & $\lambda_{\textrm{act}}$ & $ 50 \text{ agents/km}^2$ \\ 
		Holding Threshold & $\lambda_{\textrm{hold}}$ & $ 2 \text{ agents/km}^2$ \\ 
		\bottomrule
	\end{tabular}
	\vspace*{-1mm}
\end{table}

 To assess the system-level gains, we compare the proposed flux-aware asymmetric strategy against the following four baseline strategies.
	\begin{enumerate}
		\item \textit{Static Always-On:} All infrastructure modules remain active continuously. This provides an upper theoretical bound on service reliability but incurs the maximum possible energy cost.
		\item \textit{Reactive Strategy:} BSs activate strictly upon detecting local user density $\lambda \ge \lambda_{\textrm{act}}$. This serves as the fundamental baseline for analyzing boot-up latency and resultant wavefront outages.
		\item \textit{Robust Margin Strategy (RMS):} Adapted from robust optimization methods based on statistical moments and chance constraints \cite{teng2021joint}, this baseline adds a conservative spatial safety margin outside the density boundary to hedge against mobility fluctuations. In our implementation, this margin is instantiated as a fixed $8.5$ km worst-case guard ring.
		\item \textit{Traffic Snapshot Prediction (TSP) Baseline:} Inspired by learning-based traffic prediction models (e.g., LSTM \cite{wang2024base, jang2020base, azari2022energy} and DRL \cite{wu2021deep}), this serves as an oracle baseline. It assumes perfect foresight of the future static density boundary to fully cancel the temporal boot-up latency, activating BSs ahead of time based strictly on the predicted spatial density snapshots.
	\end{enumerate}

The analysis proceeds in four steps: validating the analytical coverage model and visualizing flow dynamics (Section \ref{sec:sim_visualization_setup}), validating the spatial phase-lead of the information flux (Section \ref{sec:sim_phase_lead}), quantifying the reliability gain in eliminating outages (Section \ref{sec:sim_outage}), and evaluating the overall network EE (Section \ref{sec:sim_ee}).

\subsection{Model Visualization and Analytical Validation}\label{sec:sim_visualization_setup} 

\begin{figure}[!t]
	\centering
	\vspace*{-1mm}
	\includegraphics[width=0.9\columnwidth]{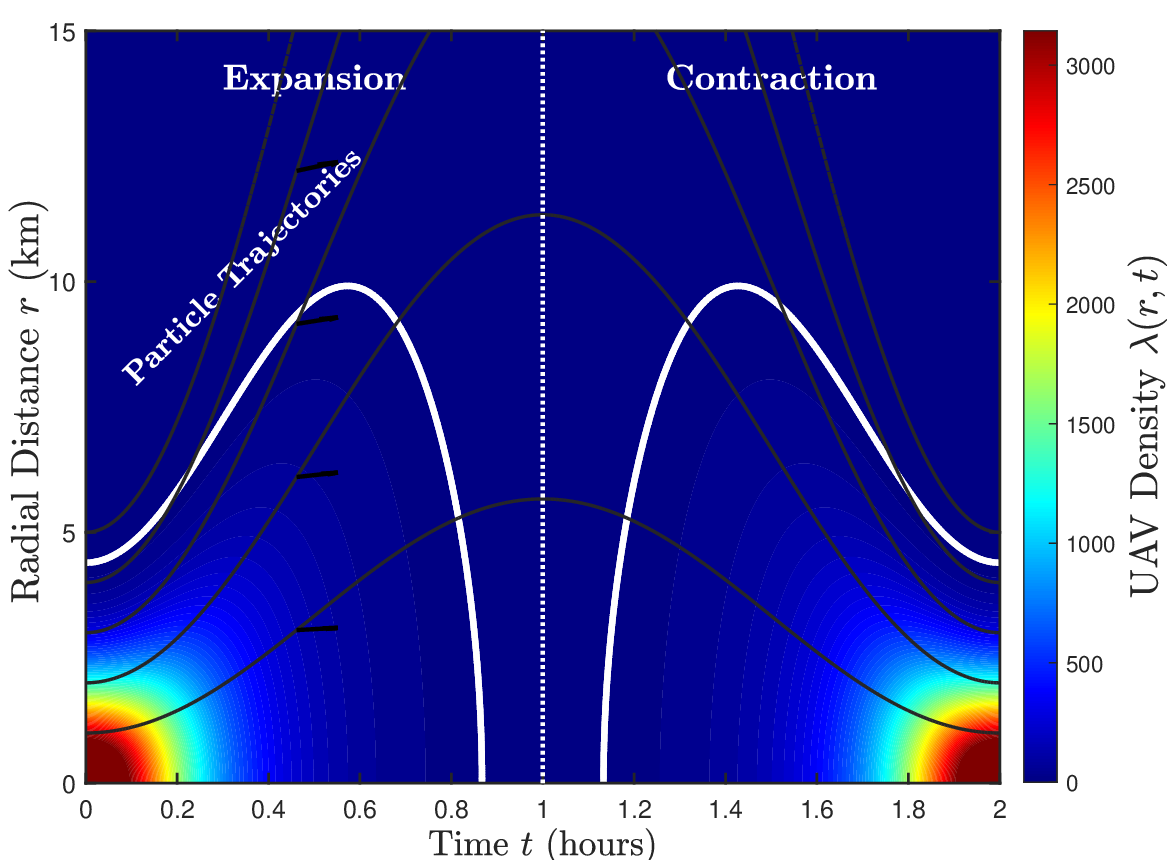} 
	\vspace*{-3.5mm}
	\caption{Spatiotemporal evolution of the pulsating logistics tide.}
	\vspace*{-4mm}
	\label{fig:spatiotemporal_flow} 
\end{figure}

\begin{figure}[!b]
	\vspace*{-4mm}
	\centering
	\includegraphics[width=0.9\columnwidth]{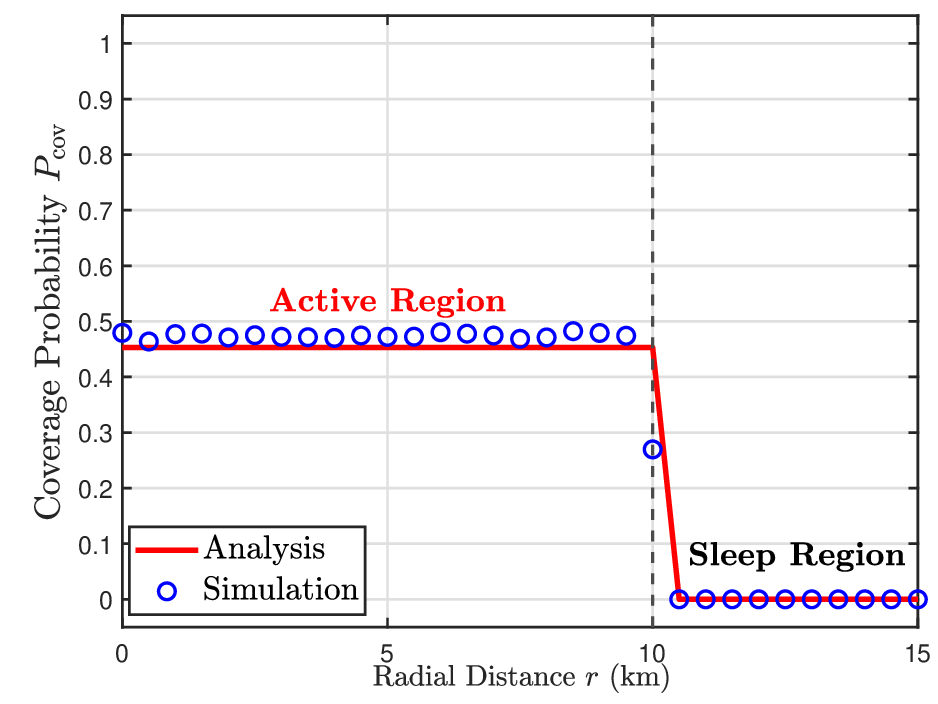} 
	\vspace*{-5mm}
	\caption{Validation of the analytical local coverage probability.}
	\label{fig:theorem_validation} 
	\vspace*{-1mm}
\end{figure}

\subsubsection{Macroscopic Dynamics Visualization}
To verify the kinematic basis of the proposed framework, Fig. \ref{fig:spatiotemporal_flow} illustrates the spatio-temporal evolution of the workload density field $\lambda(r,t)$ alongside the dynamic logistic wavefront $R_{\textrm{wf}}(t)$ represented by a white contour. By superimposing the Lagrangian characteristic trajectories of hypothetical fluid particles (plotted as thin grey lines), we observe distinct phase-dependent behaviors: during the expansion phase where $t < 1$ h, the grey trajectories diverge rapidly, corroborating that the wavefront propagation is physically driven by the macroscopic advection velocity $v_r > 0$; conversely, during the contraction phase where $t > 1$ h, the grey trajectories converge, reflecting the inward momentum as the swarm returns to the hub. This confirms that the derived information flux metric provides an accurate kinematic representation of the swarm’s momentum.

\begin{figure}[!t]
	\vspace*{-1mm}
	\centering
	\includegraphics[width=0.9\columnwidth]{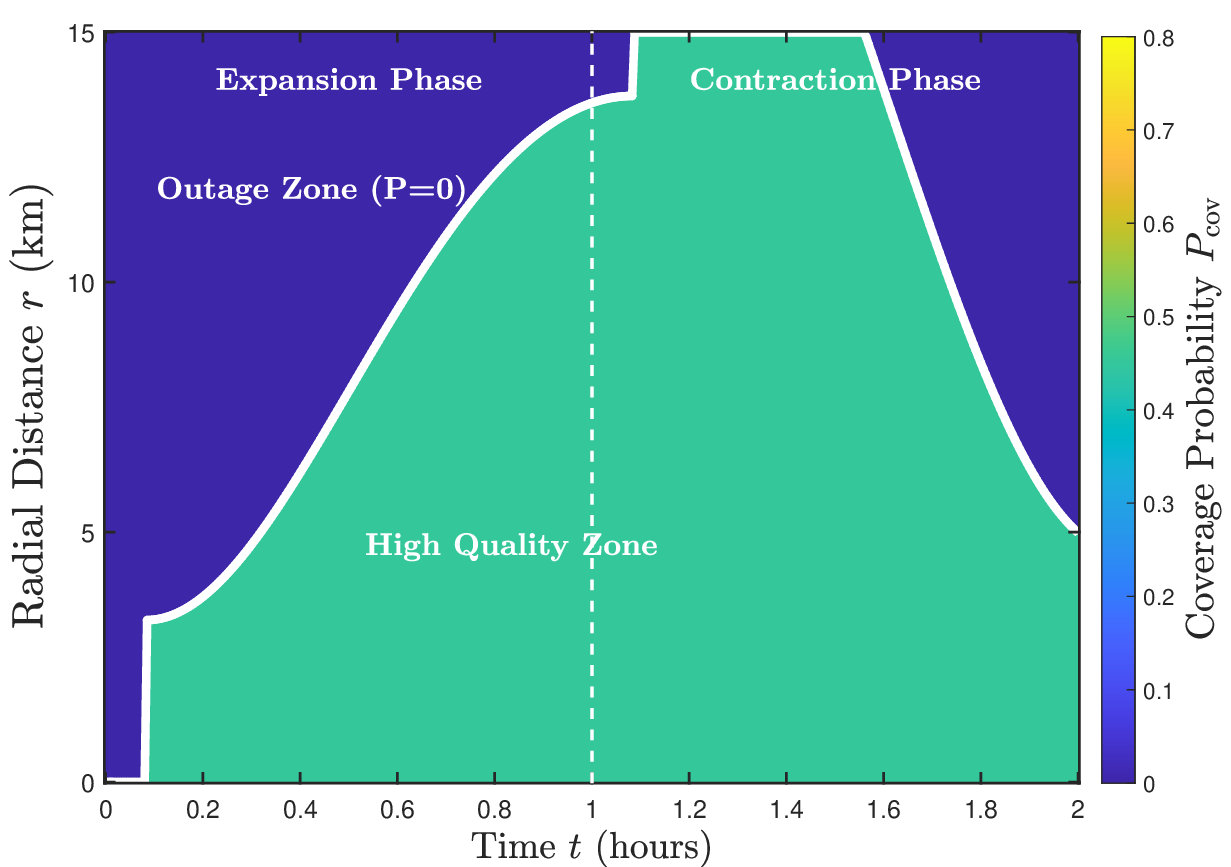} 
	\vspace*{-4mm}
	\caption{Spatiotemporal evolution of the coverage probability over a full expansion-contraction cycle.}
	\label{fig:coverage_heatmap} 
	\vspace*{-6mm}
\end{figure}

\subsubsection{Validation of Analytical Expressions}

We validate the finite network coverage model established in Theorem \ref{thm:inst_local_cov} against Monte Carlo simulations. Fig. \ref{fig:theorem_validation} presents the spatial coverage profile at a representative snapshot of $R_{\textrm{active}} = 10$ km, corresponding to the fully expanded service boundary, with path loss exponent $\nu=3.0$ and SINR threshold $\gamma=-3$ dB. Notably, the distinct boundary-induced coverage gain observed near the periphery ($r \approx 10$ km) is precisely predicted by the proposed analysis. This validation confirms that the finite interference factor derived in \eqref{eq:F_finite} correctly quantifies the reduction in aggregate interference due to the sleep mode of GBSs outside the active region, thereby eliminating the estimation error at the wavefront.

Fig. \ref{fig:coverage_heatmap} extends this validation to the temporal domain over a full expansion-contraction cycle. The heatmap reveals a traveling service plateau, visualized as the green region, which is strictly confined within the analytical activation boundary traced by the white curve. 
The precise alignment between the analytical boundary and the effective service zone confirms the robustness of the derived model under time-varying topologies. Even as the network size undergoes substantial fluctuations, the analytical framework accurately tracks the spatiotemporal evolution of the SINR field, demonstrating that the proposed fluid-based stochastic geometry offers a rigorous performance baseline for dynamic logistics corridors.

\begin{figure}[!b]
\vspace*{-5mm}
	\centering
	\includegraphics[width=0.9\columnwidth]{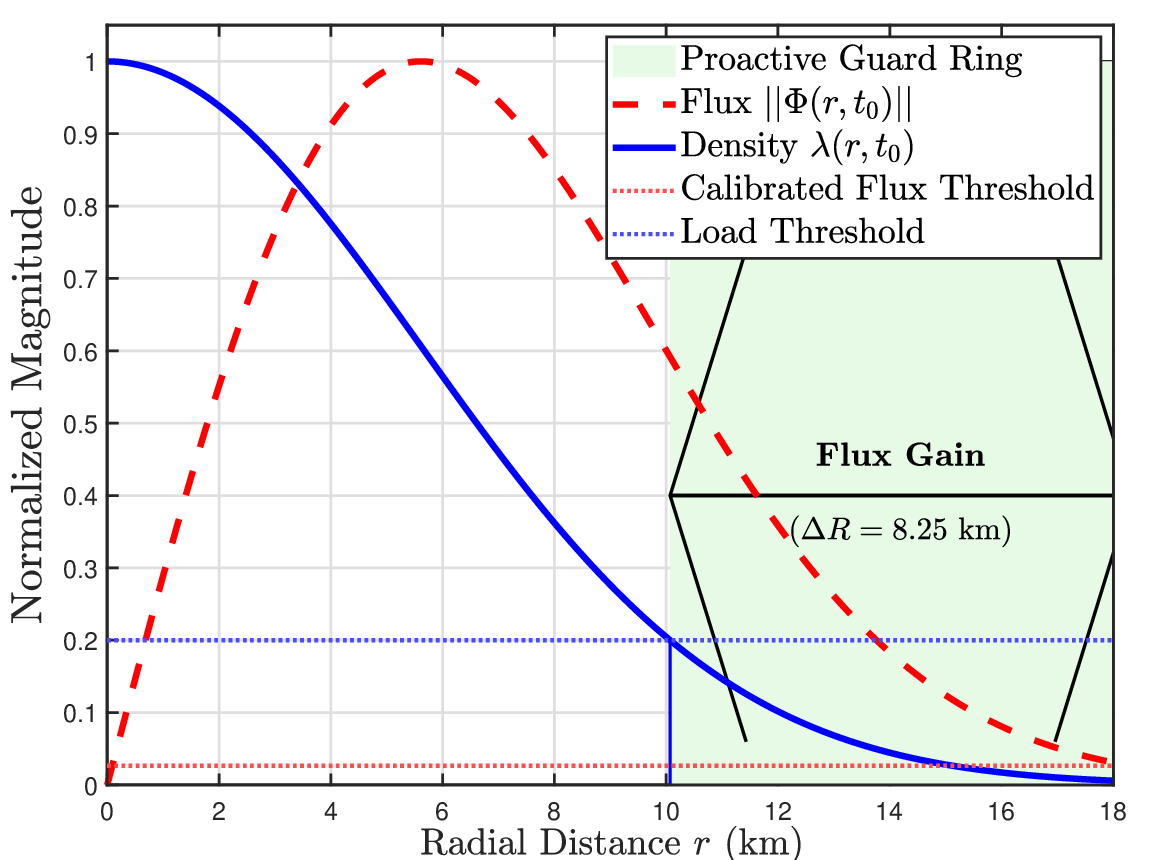}
	\vspace*{-3.7mm}
	\caption{Spatial mechanism of the flux-aware strategy at $t_0=0.6$ h, showing the proactive guard ring.}
	\vspace*{-1mm}
	\label{fig:spatial_mechanism} 
\end{figure}

\subsection{Validation of Kinematic Mechanism: Flux Phase-Lead}\label{sec:sim_phase_lead} 

We first validate the kinematic basis of the proposed framework, specifically the spatial and temporal phase-lead of the information flux $\mathbf{\Phi}$. The simulation focuses on the rapid expansion phase, as the resulting macroscopic advection velocity is the physical driver of the flux signal.

\subsubsection{Spatial Mechanism: Proactive Guard Ring}
Fig. \ref{fig:spatial_mechanism} elucidates the spatial interplay between the normalized agent density $\lambda(r, t_0)$ and the flux magnitude $\|\mathbf{\Phi}(r, t_0)\|$ at a representative snapshot $t_0\! =\! 0.6$\,h. The results reveal a fundamental limitation of conventional detection methods. As shown by the blue solid curve, the agent density decays rapidly at the network periphery. This metric is strictly coupled to the physical presence of agents: the density threshold $\lambda_{\textrm{act}}$ is crossed only upon the swarm's arrival. Consequently, any strategy relying solely on density fails to anticipate the load, making the service setup latency $\tau_{\textrm{boot}}$ an unavoidable penalty.

In contrast, the flux magnitude represented by the red dashed curve exhibits a distinct wavefront enveloping effect. Driven by the high macroscopic velocity $v_r\! \propto\! r$ in the far field, the flux signal remains significant even in regions where the agent density is negligible, specifically at distances exceeding 10\,km. This kinematic property effectively decouples the trigger signal from the physical load. As a result, the GBSs are activated by the flux condition $\|\mathbf{\Phi}\| \ge \delta_{\textrm{th}}$ well before the density condition is met. This predictive mechanism generates a substantial proactive guard ring of width $\Delta R \approx 8.25$\,km, buying the necessary time for the infrastructure to wake up before the agents actually arrive.

\begin{figure}[!b]
	\vspace*{-4mm}
	\centering
	\includegraphics[width=0.92\columnwidth]{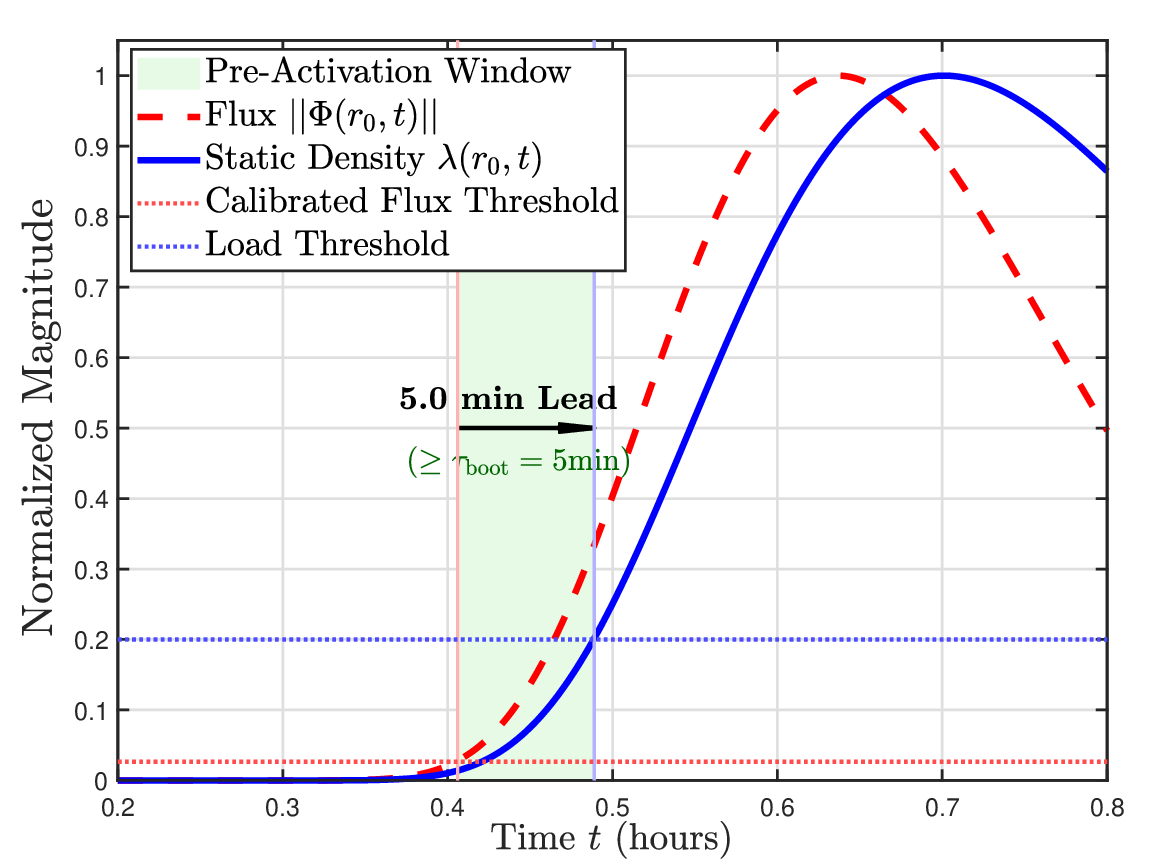} 
		\vspace*{-3mm}
	\caption{Temporal validation of zero-latency triggering at the network edge ($r=14.5$ km).}
		\vspace*{-1mm}
	\label{fig:temporal_validation} 
\end{figure}

\subsubsection{Temporal Validation: Zero-Latency Triggering}
The translation of this spatial gain into an effective temporal lead is verified in Fig. \ref{fig:temporal_validation}, which traces the metric evolution at a representative far-field location located at a radial distance of 14.5\,km. As the logistics tide propagates outward, the flux curve exhibits a sharp rise well in advance of the density curve. The calibrated trigger activates at time $t_{\textrm{flux}}$, which precedes the density arrival time $t_{\textrm{den}}$ by approximately 5 minutes. This interval closely approximates the mandatory service setup latency $\tau_{\textrm{boot}}$, confirming that the information flux functions as a robust precursor signal to ensure the GBS completes its activation sequence just as the physical load arrives.

\begin{figure}[!b]
\vspace*{-5mm}
	\centering
	\includegraphics[width=1\columnwidth]{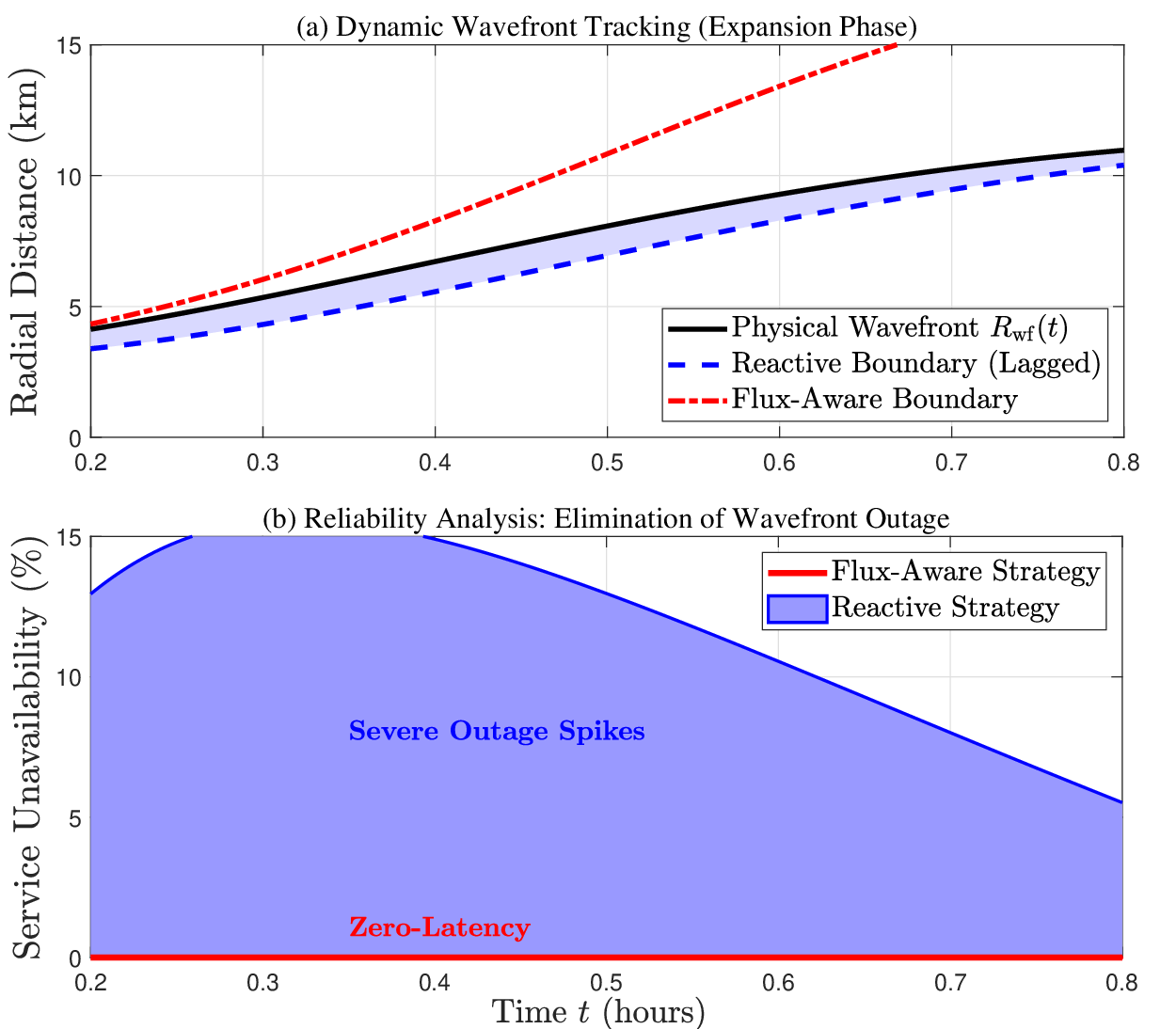}
	\vspace*{-7mm}
	\caption{Reliability analysis during the expansion phase. (a) Service-boundary trajectory versus the physical wavefront. (b) Instantaneous service unavailability.}
	\label{fig:outage_analysis} 
\vspace*{-1mm}
\end{figure}

\subsection{Reliability Analysis: Elimination of Wavefront Outage}\label{sec:sim_outage} 

Having verified the kinematic phase-lead mechanism, we now quantify the system-level reliability gain during the critical expansion phase spanning the interval $t\! \in\! [0.2, 0.8]$\,h. The primary objective is to demonstrate how the proposed strategy mitigates service outages caused by a service setup latency of $\tau_{\textrm{boot}} = 5$ min. We evaluate reliability using the service unavailability metric $\mathcal{O}(t)$ defined in  \eqref{eq:outage_metric}.

The effective service boundary is physically constrained by the activation latency, formulated as $R_{\textrm{active}}(t) = R_{\textrm{trig}}(t - \tau_{\textrm{boot}})$. For the reactive strategy, the trigger follows the static density wavefront where $R_{\textrm{trig}} = R_{\textrm{act}}$. Conversely, for the flux-aware strategy, we employ the robust calibration derived in Proposition \ref{prop:threshold_calibration} by setting $\delta_{\textrm{th}} \approx 0.2 \cdot \delta_{\textrm{max}}$ to create a sufficient safety margin against the expansion velocity.

Fig. \ref{fig:outage_analysis} visualizes the dynamic interaction between the traffic wavefront and the service boundary. The upper panel reveals a distinct tracking disparity where the reactive strategy, plotted as a blue dashed line, consistently lags behind the physical wavefront. This hysteresis creates a widening outage zone of approximately 1--2 km, confirming that static density detection is fundamentally too slow for high-mobility logistics corridors. In sharp contrast, the flux-aware strategy represented by the red dashed line maintains a service boundary that strictly envelopes the physical wavefront, effectively predicting the expansion trajectory.

The consequences are quantified in the lower panel. The reactive scheme suffers from a persistent service unavailability peaking at approximately 15\%. This value represents the complete disconnection of the leading swarm, the most strategic assets executing long-range missions. The proposed strategy reduces this outage metric to near zero, validating the effectiveness of flux-based proactive control in eliminating wavefront outages.

\subsection{System-Level Energy Efficiency and Pareto Optimization}\label{sec:sim_ee} 

Finally, we evaluate the comprehensive system performance to validate the economic viability of the proposed framework.

\subsubsection{Comparison of Raw versus Effective Energy Efficiency}
Fig. \ref{fig:system_performance} presents the normalized performance comparison across the five evaluated strategies. An initial observation is the counter intuitive behavior of the raw EE ($\eta_{\textrm{EE}}$). The reactive strategy and the TSP achieve the highest nominal $\eta_{\textrm{EE}}$. However, this metric is deceptive in high-mobility logistics scenarios. As indicated by the served traffic ratio, the reactive strategy captures only 79.6\% of the computational demand, while TSP, despite perfectly predicting the density boundary, only marginally improves this to 82.8\%. This implies that both strategies fundamentally fail to serve the critical wavefront of the digital tide---specifically, the high-velocity leading UAVs entering new coverage areas before the local density reaches the activation threshold. Consequently, their high raw energy efficiencies are merely artifacts of severe service denial, which is operationally unacceptable for mission-critical logistics.

To resolve this efficiency-reliability trade-off, we analyze the effective EE ($\eta_{\textrm{eff}}$), which explicitly penalizes service outages. Under this QoS-constrained metric, the performance of both the reactive and TSP strategies collapses. To guarantee service availability, the RMS enforces a static worst-case spatial guard ring, successfully elevating the served traffic to approximately 99.0\%. However, this rigid geometric expansion incurs massive energy waste through severe over-activation, causing its $\eta_{\textrm{eff}}$ to degrade to a level comparable with the static always-on upper bound. In contrast, the proposed flux-aware strategy leverages the fluid-dynamic velocity field to dynamically adjust the activation boundary. It seamlessly captures the wavefront momentum, achieving a 99.1\% traffic service rate that is numerically comparable to the always-on benchmark, while promptly shrinking the boundary during swarm contraction to reduce trailing energy waste. By inherently balancing kinematic reliability and dynamic power consumption, the flux-aware framework achieves the maximum effective EE among all the evaluated approaches.

\begin{figure}[!t]
		\vspace*{-2mm}
	\centering
	\includegraphics[width=\columnwidth]{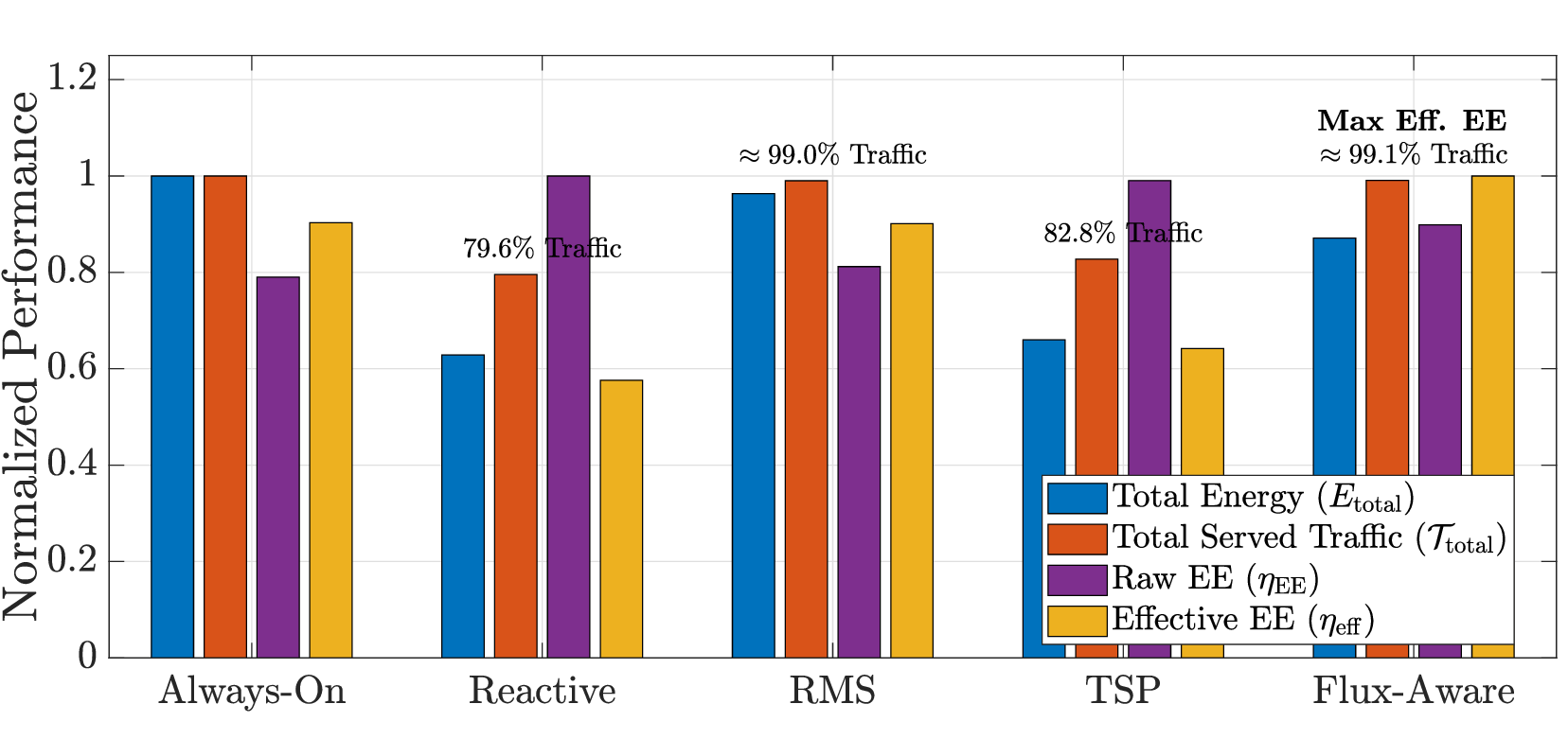} 
	\vspace*{-8mm}
	\caption{System-level performance comparison across benchmark strategies.}
	\label{fig:system_performance} 
	\vspace*{-4mm}
\end{figure}

\begin{figure}[!t]
\vspace*{-1mm}
	\centering
	\includegraphics[width=1\columnwidth]{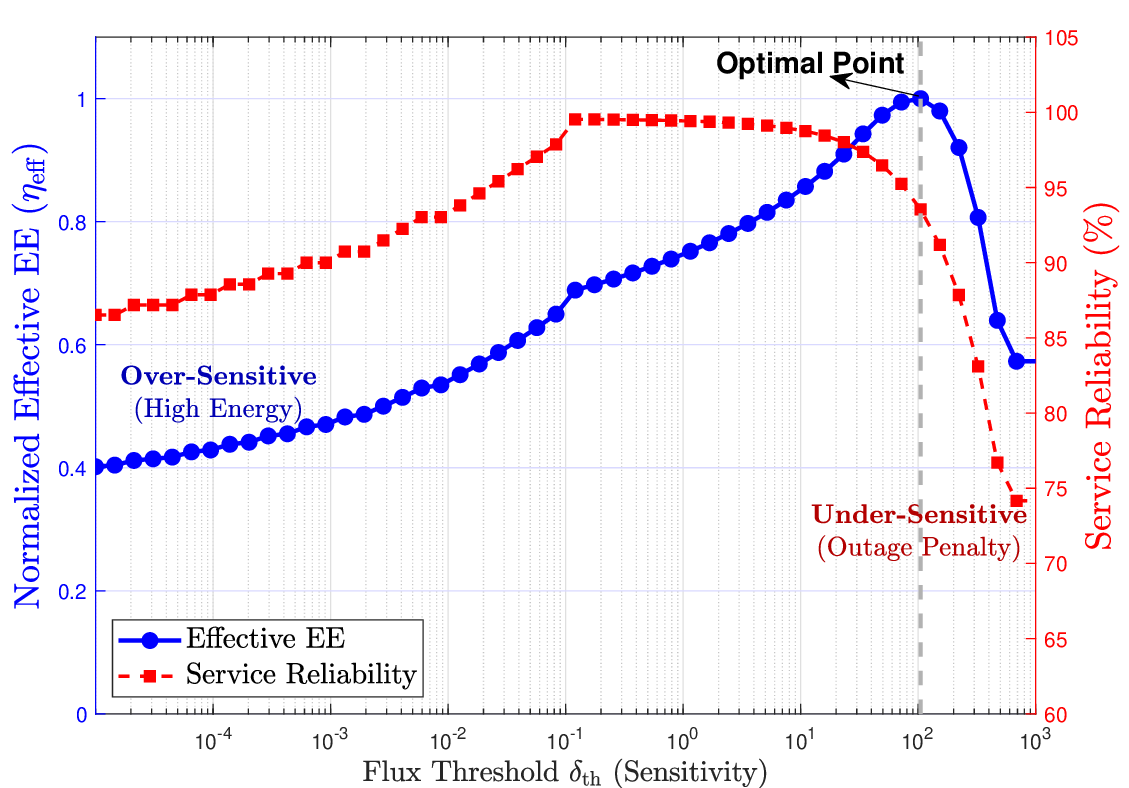}
	\vspace*{-7mm}
	\caption{Impact of the flux sensitivity threshold $\delta_{\textrm{th}}$ on system performance and the resulting Pareto-optimal operating point.}
	\label{fig:ee_optimization} 
	\vspace*{-6mm}
\end{figure}

\subsubsection{Optimal Threshold Calibration}

To validate the existence of a unique optimal operating point, Fig. \ref{fig:ee_optimization} analyzes the sensitivity of system performance to the flux threshold $\delta_{\textrm{th}}$. The effective EE and service reliability curves depicted in Fig.~\ref{fig:ee_optimization} elucidate the fundamental trade-off between energy consumption and service reliability. Specifically, the optimization landscape reveals how flux sensitivity dictates the balance between energy waste and outage risks.

\textit{2.a)~Over-Sensitive Regime ($\delta_{\text{th}}\! <\! 10^{-1}$):} 
Here, the low threshold results in high sensitivity. Both the effective EE and service reliability improve as $\delta_{\text{th}}$ increases. As $\delta_{\text{th}}$ approaches $10^{-1}$, the service reliability reaches approximately 100\% but this comes at a severe energy cost with the effective EE suppressed below 0.7. This is because the system maintains an excessively wide proactive guard ring, activating GBSs far earlier than necessary.
		
\textit{2.b)~Transition Regime ($10^{-1}\! \le\! \delta_{\text{th}}\! <\! 10^{2}$):}
This interval marks the efficiency-gaining phase. As $\delta_{\text{th}}$ increases, the system progressively sheds spatial redundancy, and the effective EE climbs steadily from 0.7 to its peak, while reliability remains largely robust above 95\%. This confirms that the flux trigger still provides sufficient lead time to compensate for $\tau_{\text{boot}}$ while eliminating the waste of an oversized guard ring.
		
\textit{2.c)~Pareto-Optimal Point ($\delta_{\text{th}}\! \approx\! 10^{2}$):} 
The global maximum of the effective EE curve defines the optimal equilibrium. Crucially, this peak efficiency does not correspond to perfect reliability: at the optimal $\delta_{\text{th}}$, service reliability drops from its peak of 100\% to approximately 93\%.
		
\textit{2.d)~Under-Sensitive Regime ($\delta_{\text{th}}\! >\! 10^{2}$):}
Beyond the optimal point, $\delta_{\text{th}}$ becomes excessively high. Physically, the proactive guard ring contracts until its width is insufficient to compensate for the service setup latency $\tau_{\text{boot}}$. In
the asymptotic limit, the flux trigger is effectively disabled, causing the strategy to degenerate into the reactive baseline. Consequently, the network detects the wavefront too late, causing service reliability to plummet below 75\%. The heavy wavefront outage triggers the steep slope of the QoS penalty function $\mathcal{U}_{\beta}$, driving the effective EE to collapse toward the negligible performance floor of the reactive strategy.

This result demonstrates a strategic trade-off: to maximize the effective EE, the system must tolerate a marginal degradation in reliability, corresponding to an approximately 7\% outage risk at the wavefront. This calibration tightens the proactive guard ring to the minimum size necessary to offset hardware latency, thereby locating the true Pareto-optimal frontier that is unattainable by strictly reactive mechanisms.

	\subsection{System Sensitivity and Robustness Analysis}\label{sec:sim_sensitivity} 
	
	\subsubsection{Sensitivity to Service Setup Latency}
	We evaluate the system robustness against service setup latency $\tau_{\mathrm{boot}} \in [1, 10]$ min, representing varying MEC instantiation and RF wake-up overheads. We fix the flux sensitivity threshold to the Pareto-optimal $\delta_{\mathrm{th}}\! =\! 10^2$. As shown in Fig. \ref{fig:sensitivity}\,(a), the reactive baseline's service unavailability scales linearly with $\tau_{\mathrm{boot}}$, approaching 37\% at a 10-minute latency. Conversely, the flux-aware strategy exhibits graceful degradation, restricting the outage to 6.2\% at 5 minutes and below 15\% at 10 minutes. This confirms that the proactive guard ring effectively absorbs the temporal lag. Consequently, as depicted in Fig. \ref{fig:sensitivity}\,(b), the reactive strategy's effective EE plummets due to severe QoS penalties, whereas the proposed framework maintains superior energy efficiency across the latency spectrum, validating its stability for practical deployments.
	
	\begin{figure}[!t]
		\centering
		\includegraphics[width=\columnwidth]{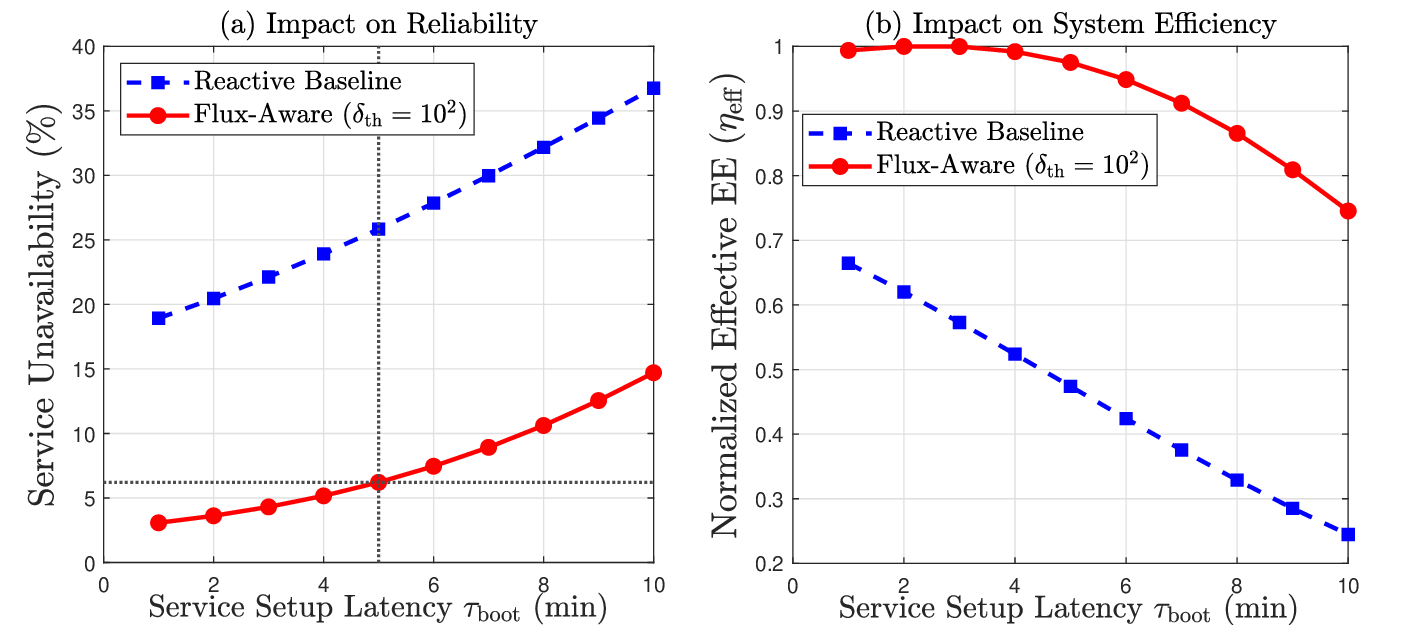}
		\vspace*{-7mm}
		\caption{Sensitivity of system performance to the service setup latency under the Pareto-optimal flux threshold: (a) impact on service reliability, and (b) impact on effective energy efficiency.}
		\label{fig:sensitivity} 
		\vspace*{-4mm}
	\end{figure}
	
	\subsubsection{Robustness Against RF Constraints}
	We further investigate the system sensitivity to underlying RF parameters under five scenarios using the control variable method, as illustrated in Fig. \ref{fig:rf_robustness}. First, fixing $\eta_{\mathrm{SE}}\! =\! 2.0$\,bps/Hz and increasing $\gamma$ from $-5$\,dB to $5$\,dB nonlinearly degrades the coverage probability, compressing the effective EE curves downward. Second, fixing $\gamma\! =\! -5$\,dB and increasing $\eta_{\mathrm{SE}}$ from $2.0$ to $4.0$\,bps/Hz linearly multiplies the spatial capacity, stretching the effective EE curves upward. 
	
	\begin{figure}[!b]
		\vspace*{-6mm}
		\centering
		\includegraphics[width=0.95\columnwidth]{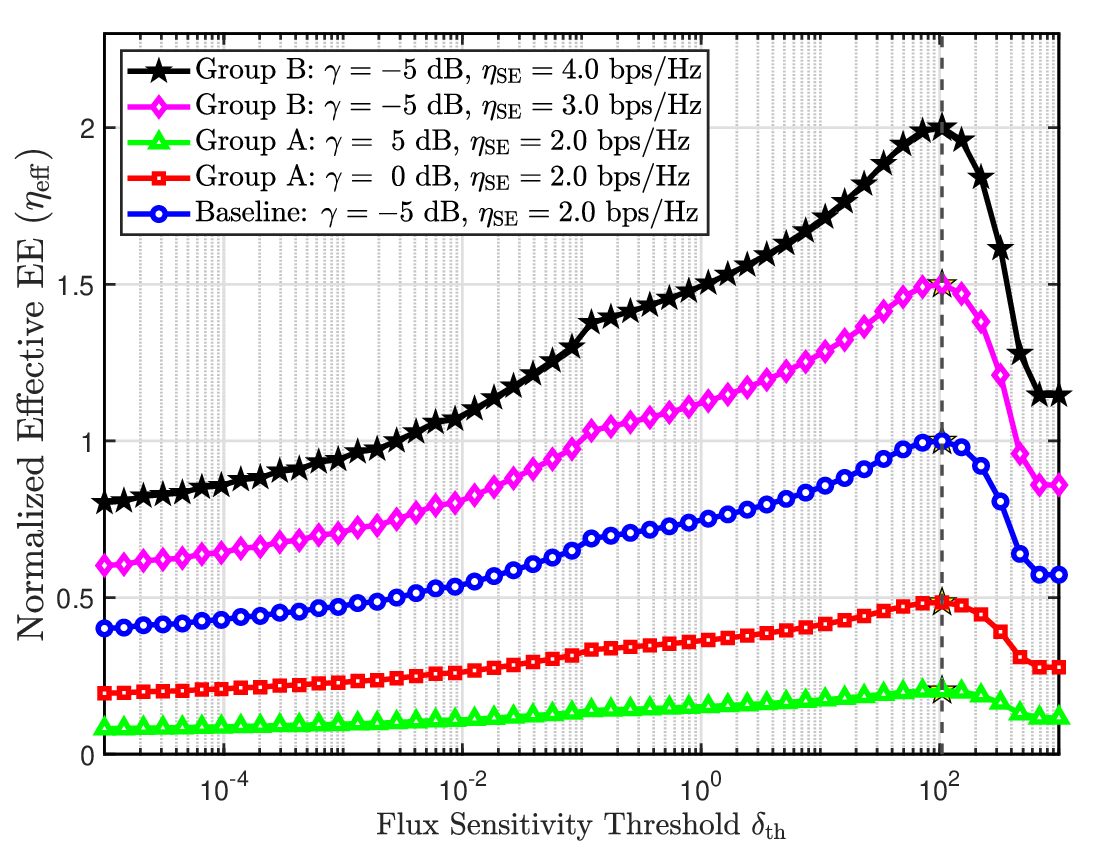}
		\vspace*{-3mm}
		\caption{Normalized effective EE versus the flux sensitivity threshold under varying SINR thresholds and spectral efficiency requirements.}
		\label{fig:rf_robustness} 
		\vspace*{-1mm}
	\end{figure}
	
	Despite these severe vertical fluctuations, the Pareto-optimal flux sensitivity threshold remains stationary at $\delta_{\mathrm{th}} \approx 10^2$ across all scenarios. This invariant optimal point suggests that the required proactive guard ring is primarily governed by the macroscopic expansion velocity and hardware setup latency, while being largely insensitive to the tested microscopic RF link conditions. Consequently, network operators may be able to upgrade modulation and coding schemes without materially recalibrating the kinematic infrastructure triggering logic, although the absolute performance levels may still change.

	\section{Discussion and Practical Limitations}\label{sec:discussion} 
	
	While the proposed fluid-dynamic modeling and flux-aware activation framework establish a tractable methodology for macroscopic infrastructure provisioning, certain idealized mathematical abstractions warrant further discussion regarding their practical limitations under non-ideal conditions.
	
	\subsubsection{Impact of Spatial Deployment Deviations} 
	The theoretical framework evaluates instantaneous coverage probability and energy efficiency under the assumption of a homogeneous PPP, denoted by $\Phi_{\mathrm{BS}}$. In complex urban environments, base station deployments are constrained by topography and zoning laws. Stochastic geometry literature establishes that the complete spatial randomness of a PPP typically yields a conservative lower bound on coverage probability compared to planned repulsive networks. Consequently, if the actual cellular deployment exhibits a more regular grid-like pattern, the variance of the nearest-serving distance decreases, mitigating severe local interference. In such scenarios, the effective service availability improves beyond our analytical predictions. Conversely, extreme clustering of infrastructure may introduce localized coverage holes. However, the core mechanism proposed in this paper, namely, the kinematic phase-lead of the macroscopic information flux over the scalar density, is expected to remain qualitatively robust. This predictive capability is governed by the continuity equation of the fluid dynamics and is decoupled from the microscopic spatial distribution of the ground infrastructure.
	
	\subsubsection{Fluid Approximation Boundaries and Abrupt Reconfiguration} 
Treating the UAV swarm as a continuous compressible fluid is effective under macroscopic high-density conditions with smooth workload flows. However, the validity of this continuous fluid model relies on the cohesive movement of the logistics fleet. Under non-ideal conditions, such as abrupt swarm reconfiguration or chaotic individual trajectories, the macroscopic continuity equation loses its predictive accuracy. In these regimes, the macroscopic flux vector fails to resolve microscopic discontinuities, causing the proactive system performance to degrade toward the reactive baseline. Addressing this limitation necessitates a hybrid methodology combining macroscopic fluid dynamics with discrete multi-agent tracking in future research.
	
	\subsubsection{Impact of Estimation Noise} 
	The current analysis assumes the accurate acquisition of macroscopic state variables. In practical deployments, sensory measurements of the swarm density and velocity fields are subject to estimation noise. Such inaccuracies in the flux estimation can induce localized false alarms or delayed infrastructure activations. Mitigating this vulnerability requires the integration of temporal smoothing techniques or predictive filters, such as Kalman filtering, prior to the flux evaluation to ensure robust activation triggers.
	
	\subsubsection{Heterogeneous Service Startup Delays} 
The theoretical and simulation analyses evaluate system robustness under uniform service setup latencies. Real-world infrastructure often exhibits heterogeneous startup times due to varying hardware capabilities and backhaul states across different cells. Managing such spatial heterogeneity requires extending the current framework by transforming the global scalar flux sensitivity threshold, denoted by $\delta_{\mathrm{th}}$, into a spatially varying field. This adaptation allows the system to dynamically adjust the proactive guard ring on a per-cell basis according to localized hardware constraints.

	\subsubsection{Applicability to Geometry-Constrained Corridors}
	The current analytical evaluation assumes a radially symmetric Gaussian tide. In practice, UAV logistics often follow geometry-constrained flight corridors due to urban topology or airway regulations. Nevertheless, the proposed flux-induced phase-lead mechanism is expected to remain valid under such asymmetric or one-dimensional constrained trajectories. This resilience is fundamentally rooted in the continuity equation. For a 1D delivery corridor, the spatial distribution reduces to a longitudinal profile $\lambda(x,t)$. Driven by the conservation of mass during the expansion phase, the macroscopic advection velocity $v_x$ still scales proportionally with the propagation distance $x$. Consequently, the information flux magnitude $\|\mathbf{\Phi}\| = \lambda \cdot v_x$ inherently decays slower than the scalar density $\lambda$ at the wavefront. Thus, the spatial monotonicity as proven in Lemma \ref{lemma:flux_ratio} mathematically holds, ensuring that the flux trigger continues to generate a proactive temporal margin to compensate for service setup latency in corridor-based deployments.

\section{Conclusion}\label{sec:conclusion} 

This paper has investigated the critical conflict between EE and service reliability in urban UAV logistics networks. We have shown that conventional reactive sleeping strategies yield misleadingly high nominal efficiency, a metric inadvertently resulting from the systematic shedding of the energy-intensive traffic wavefront. To resolve this misalignment, we have transcended the static snapshot limitations of classical stochastic geometry by establishing a fluid-dynamic traffic model and proposing a flux-aware asymmetric activation strategy. By exploiting the intrinsic spatial phase-lead of the information flux, our control logic generates a proactive guard ring that effectively compensates for service setup latency. Numerical results have validated that this framework achieves a zero-latency response, reducing wavefront outages from approximately 20\% to near zero compared to reactive baselines. Under our derived QoS-penalized metric, the proposed strategy achieves Pareto optimality, demonstrating that integrating macroscopic kinematic states into network control loops offers a resilient pathway for aerial mobile computing systems.

\small

	\makeatletter
	\def\@IEEEBIOskipN{0.1\baselineskip}
	\makeatother
\begin{IEEEbiography}
	[{\includegraphics[width=1in,height=1.25in,clip,keepaspectratio]{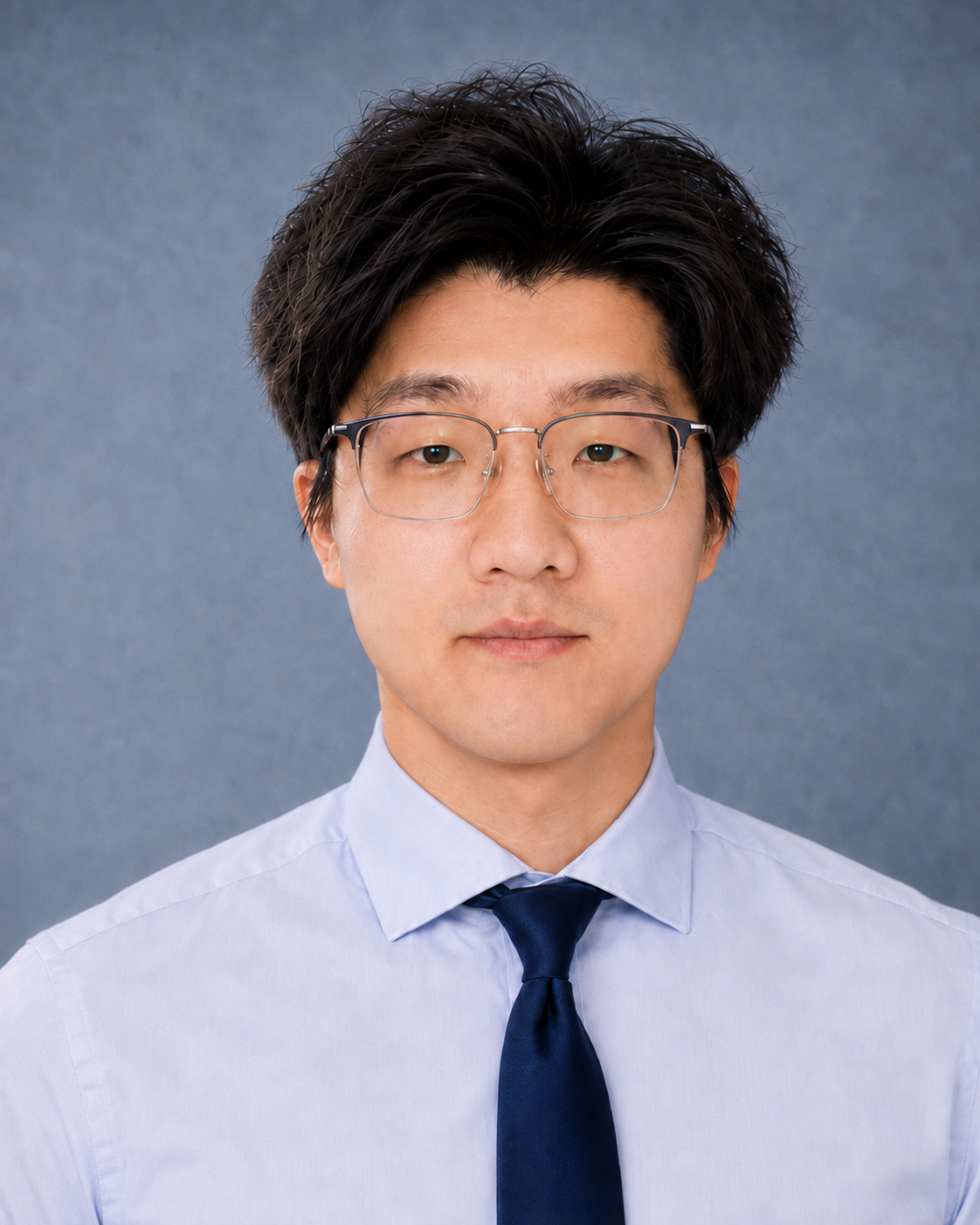}}]
	{Wen-Yu Dong}
received the B.S. degree in electronic and information engineering from Sichuan University, Chengdu, China, in 2019, and the Ph.D. degree in information and communication engineering from the School of Information and Communication Engineering, Beijing University of Posts and Telecommunications, Beijing, China, in 2025. He is currently a researcher with the China Telecom Research Institute, Beijing, China. His current research interests include space-air-ground integrated networks, information theory, stochastic geometry, and wireless communications.
\end{IEEEbiography}

\begin{IEEEbiography}[{\includegraphics[width=1in,height=1.25in,clip,keepaspectratio]{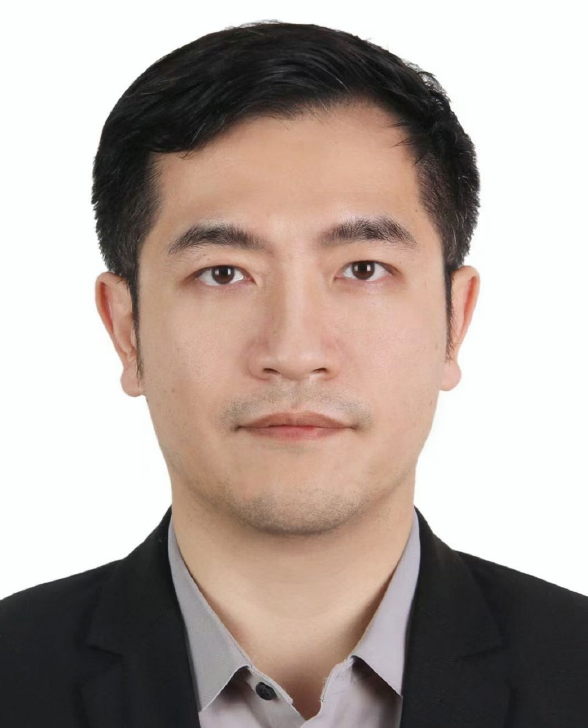}}]{Song Zhao}
 is a senior engineer with the Future Technology Center, China Telecom Research Institute, Beijing, China. He received his Ph.D. degree in Telecommunications and Information Systems from Beijing University of Posts and Telecommunications, Beijing, China. He serves as a delegate of China Telecom in the 3GPP SA2 and SA5 working groups, and as rapporteur for multiple 3GPP work items related to network automation and intelligence. His research interests include wireless channel modeling, radio resource management, proximity services, and network AI.
\end{IEEEbiography}

\begin{IEEEbiography}
	[{\includegraphics[width=1in,height=1.25in,clip,keepaspectratio]{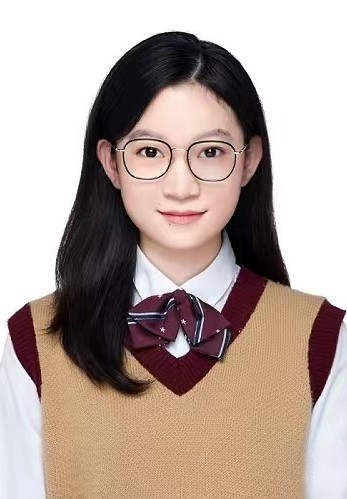}}]
	{Rui-Si Han}
	received her B.S. degree in E-Commerce and Law from Beijing University of Posts and Telecommunications, Beijing, China, in 2021, and her M.Eng. degree in Information and Communication Engineering from the same university in 2024. She is currently a researcher at the China Telecom Cloud Network Operating System R\&D Center, Beijing, China. Her research interests include mobile ad hoc networks and research project management.
\end{IEEEbiography}

\begin{IEEEbiography}[{\includegraphics[width=1in,height=1.25in,clip,keepaspectratio]{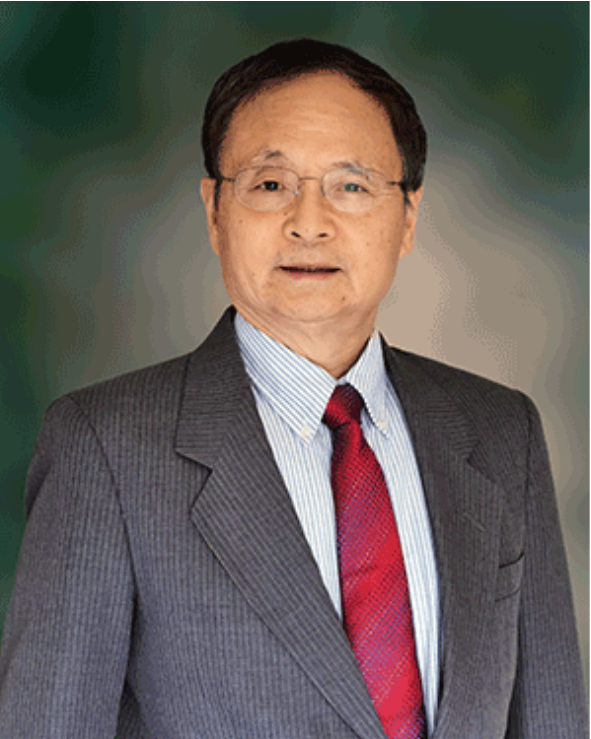}}]{Qi Bi}
	(Fellow, IEEE) is Chief Scientist of China Telecom and CTO of its Research Institute, specializing in 5G and 6G. He received his M.S. from Shanghai Jiao Tong University and Ph.D. from Pennsylvania State University. He was awarded Bell Labs Fellow in 2002, the Bell Labs President’s Gold Awards in 2000 and 2002, the Asian American Engineer of the Year in 2005, the Beijing Outstanding Contribution Award for Innovation and Entrepreneurship for Overseas Scholars in 2019, and three Patent Silver Prizes from the China National Intellectual Property Administration from 2022 to 2024.
\end{IEEEbiography}

\begin{IEEEbiography}[{\includegraphics[width=1in,height=1.25in,clip,keepaspectratio]{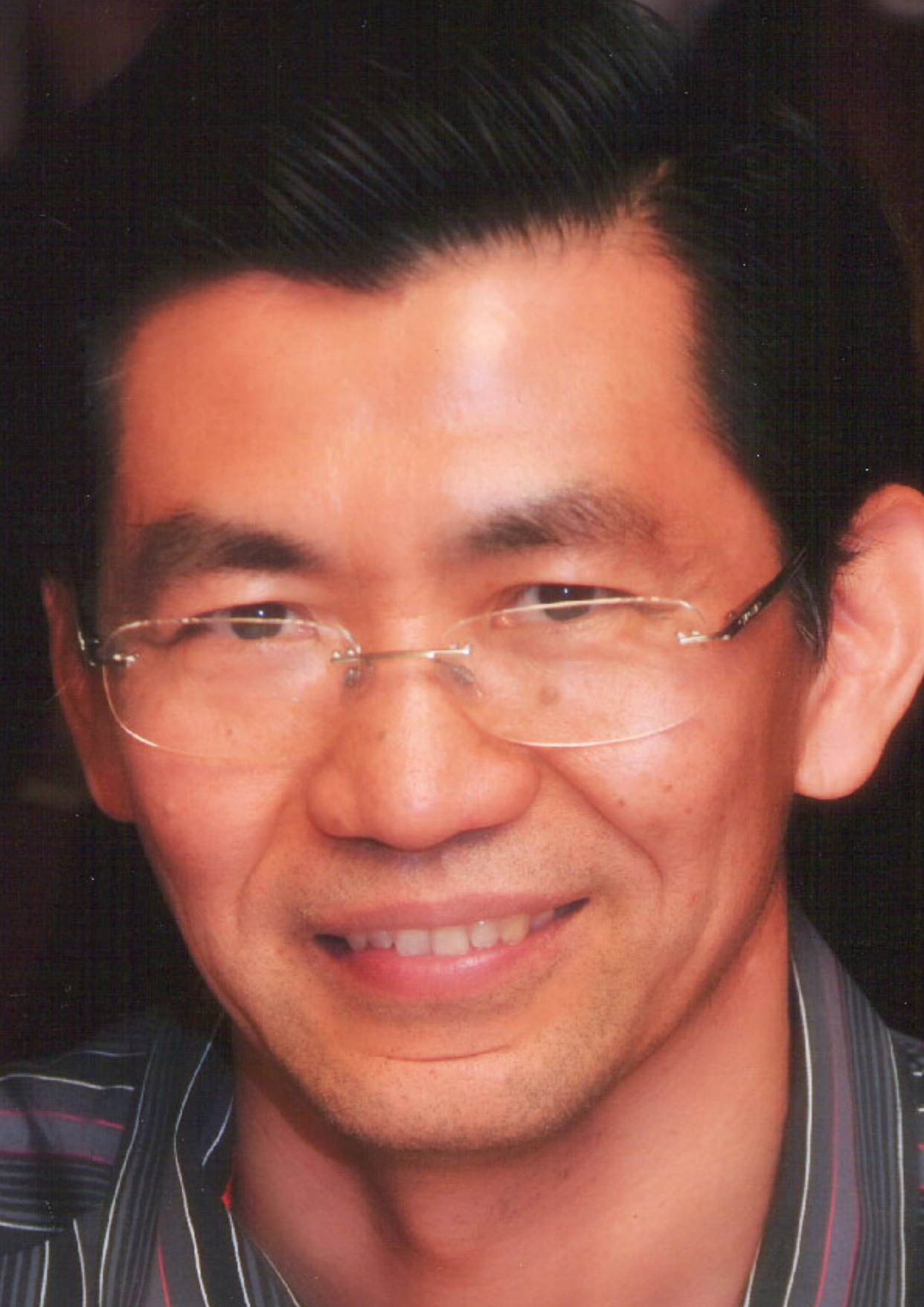}}]{Sheng Chen}
	(Life Fellow, IEEE) received his BEng degree from the East China Petroleum Institute, Dongying, China, in 1982, and his PhD degree from the City University, London, in 1986, both in control engineering. In 2005, he was awarded the higher doctoral degree, Doctor of Sciences (DSc), from the University of Southampton, Southampton, UK. From 1986 to 1999, He held research and academic appointments at the Universities of Sheffield, Edinburgh and Portsmouth, all in UK. Since 1999, he has been with the School of Electronics and Computer Science, the University of Southampton, UK, where he holds the post of Professor in Intelligent Systems and Signal Processing. Dr Chen's research interests include adaptive signal processing, wireless communications, neural network and machine learning. He has published over 700 research papers. Professor Chen has 22,000+ Web of Science citations with h-index 65 and 43,000+ Google Scholar citations with h-index 87. Dr. Chen is a Fellow of the United Kingdom Royal Academy of Engineering, a Fellow of Asia-Pacific Artificial Intelligence Association and a Fellow of IET. He is one of the original 200 ISI highly cited researchers in engineering (March 2004). 
Professor Chen is IEEE ComSoc Signal Processing and Computing for Communications (SPCC) 2025 Technical Recognition Award recipient. 	
\end{IEEEbiography}

\end{document}


\title{\LARGE Supplementary Material for 
	``Digital Tides: A Fluid-Dynamic Framework for Flux-Aware Infrastructure Provisioning in UAV Logistics Networks''}

	\IEEEaftertitletext{\vspace{-8mm}}
	\maketitle

\setcounter{equation}{0}
\setcounter{figure}{0}
\setcounter{table}{0}

\section{Derivation of Macroscopic Velocity Field}\label{app:velocity_derivation} 

The radial velocity $v_r(r,t)$ is determined by the continuity equation (4) in polar coordinates:
\begin{equation} 
	\frac{\partial \lambda}{\partial t} + \frac{1}{r} \frac{\partial}{\partial r} (r \lambda v_r) = 0.
\end{equation}
Integrating from $0$ to $r$ with boundary condition $v_r(0,t)=0$, we obtain:
\begin{equation} 
	r \lambda(r,t) v_r(r,t) = - \int_0^r x \frac{\partial \lambda(x,t)}{\partial t} \mathrm{d}x.
	\label{eq:app_integral_def}
\end{equation}
For the Gaussian field $\lambda(x,t) = \frac{N(t)}{2\pi \sigma^2(t)} e^{-\frac{x^2}{2\sigma^2(t)}}$, we first determine its temporal derivative. Differentiating $\ln \lambda$ with respect to time $t$, and denoting the temporal derivative by the dot notation, i.e., $\dot{x} \triangleq \partial x / \partial t$, we obtain:
\begin{equation}\label{eq:TempDer} 
	\frac{1}{\lambda} \!\frac{\partial \lambda}{\partial t} \!= \!\frac{\partial}{\partial t} \! \left( \ln N \! -\!  2\ln \sigma\! - \!\frac{x^2}{2\sigma^2} \!\!\right)\! =\! \frac{\dot{N}}{N} + \frac{\dot{\sigma}}{\sigma} \left( \frac{x^2}{\sigma^2} - 2 \right).
\end{equation}
Substituting (\ref{eq:TempDer}) into \eqref{eq:app_integral_def} leads to:
\begin{equation} 
	r \lambda v_r = - \frac{\dot{N}}{N} \underbrace{\int_0^r x \lambda \mathrm{d}x}_{\mathcal{I}_{\textrm{load}}} - \frac{\dot{\sigma}}{\sigma} \underbrace{\int_0^r x \lambda \left( \frac{x^2}{\sigma^2} - 2 \right) \mathrm{d}x}_{\mathcal{I}_{\textrm{exp}}} ,
	\label{eq:app_split}
\end{equation}
which is decomposed into two distinct physical components: a cumulative load term $\mathcal{I}_{\textrm{load}}$ and a spatial expansion term $\mathcal{I}_{\textrm{exp}}$.
To simplify the spatial expansion term $\mathcal{I}_{\textrm{exp}}$, we leverage the derivative property specific to the Gaussian kernel: $\partial_x \lambda = -\frac{x}{\sigma^2}\lambda$.
Applying integration by parts to the cubic term $\int \frac{x^3}{\sigma^2}\lambda \mathrm{d}x$ with $u=x^2$ and $\mathrm{d}v = -\mathrm{d}\lambda$, we obtain:
\begin{align} 
	\int_0^r \frac{x^3}{\sigma^2}\lambda \textrm{d}x &= \!\! \int_0^r \!x^2 \!\!\left(\! -\frac{\partial \lambda}{\partial x} \right) \!\textrm{d}x 
	\!= \!-r^2 \lambda(r,t)\! +\! 2 \mathcal{I}_{\textrm{load}}.
\end{align}
Substituting this result back into the definition of $\mathcal{I}_{\textrm{exp}}$ in \eqref{eq:app_split}:
\begin{equation} 
	\mathcal{I}_{\textrm{exp}} = \left( -r^2 \lambda + 2\mathcal{I}_{\textrm{load}} \right) - 2\mathcal{I}_{\textrm{load}} = -r^2 \lambda(r,t).
\end{equation}
The integral terms $\mathcal{I}_{\textrm{load}}$ cancel perfectly, isolating the boundary term. The velocity equation simplifies to:
\begin{equation} 
	r \lambda(r,t) v_r(r,t) = - \frac{\dot{N}(t)}{N(t)} \mathcal{I}_{\textrm{load}} + \frac{\dot{\sigma}(t)}{\sigma(t)} r^2 \lambda(r,t).
\end{equation}
Note that the cumulative load integral evaluates to $\mathcal{I}_{\textrm{load}} = \frac{N(t)}{2\pi}\left( 1 - \exp\left(-\frac{r^2}{2\sigma^2(t)}\right) \right)$. Dividing both sides by $r \lambda(r,t)$, we obtain the exact macroscopic velocity:
\begin{equation} 
	v_r(r,t) = r \frac{\dot{\sigma}(t)}{\sigma(t)} - \frac{\dot{N}(t)}{N(t)} \frac{\sigma^2(t)}{r} \left( \exp\left(\frac{r^2}{2\sigma^2(t)}\right) - 1 \right).
\end{equation}
In the transport-dominated regime, where the source-variation term remains sufficiently small over the spatial region of interest, the velocity field can be effectively approximated by its leading transport component:
\begin{equation} 
	v_r(r,t) \approx r \frac{\dot{\sigma}(t)}{\sigma(t)}.
\end{equation}
This completes the derivation.\qed

\section{Proof of Proposition 1}\label{supp:proof_activation_radius}

\begin{proof}
	The effective activation radius $R_{\textrm{active}}(t)$ is determined by the phase-dependent control logic defined in (20). We explicitly derive the component radii for the density-based and flux-based triggers, respectively.
	
	\textit{1) Static Density Boundaries:}
	Consider a generic density threshold $\Lambda \in \{\lambda_{\textrm{act}}, \lambda_{\textrm{hold}}\}$. The activation boundary is the spatial locus satisfying $\lambda(r,t) = \Lambda$. Inverting the Gaussian intensity function (1) with respect to $r$ yields:
	\begin{equation} 
		\frac{N(t)}{2\pi \sigma^2(t)} \exp\left( - \frac{r^2}{2\sigma^2(t)} \right) = \Lambda.
	\end{equation}
	Rearranging the terms and taking the natural logarithm, we obtain the squared radius:
	\begin{equation} 
		r^2 = 2\sigma^2(t) \ln \left( \frac{N(t)}{2\pi \sigma^2(t) \Lambda} \right).
	\end{equation}
	Imposing the physical constraint that the radius must be non-negative, i.e., existing only if the peak density exceeds $\Lambda$, we derive the general closed-form expression:
	\begin{equation} 
		R_{\textrm{den}}(t; \Lambda) = \sigma(t) \sqrt{2 \left[ \ln \left( \frac{N(t)}{2\pi \sigma^2(t) \Lambda} \right) \right]^+ }.
	\end{equation}
	Substituting $\Lambda = \lambda_{\textrm{act}}$ and $\Lambda = \lambda_{\textrm{hold}}$ yields the expressions for $R_{\textrm{act}}(t)$ in (22a) and $R_{\textrm{hold}}(t)$ in (22b), respectively. Since the logarithmic function is strictly monotonic and $\lambda_{\textrm{hold}} < \lambda_{\textrm{act}}$, it strictly follows that $R_{\textrm{hold}}(t) \ge R_{\textrm{act}}(t)$.
	
	\textit{2) Flux Trigger Radius:}
	The flux boundary satisfies the condition $\|\mathbf{\Phi}(r,t)\|\! =\! \delta_{\textrm{th}}$. Substituting the macroscopic velocity $v_r\! \approx\! r |\dot{\sigma}|/\sigma$ into the flux definition $\mathbf{\Phi}\! =\! \lambda v_r \mathbf{e}_r$, we obtain:
	\begin{equation} 
		\left( \frac{N(t)}{2\pi \sigma^2(t)} \exp\left(-\frac{r^2}{2\sigma^2(t)}\right) \right) \cdot \left( r \frac{|\dot{\sigma}(t)|}{\sigma(t)} \right) = \delta_{\textrm{th}}.
	\end{equation}
	Grouping the spatially dependent terms on the left-hand side, we define the spatial profile function $g(r)$:
	\begin{equation} 
		g(r)\! \triangleq\! \frac{r}{\sigma^3(t)} \exp\left( - \frac{r^2}{2\sigma^2(t)} \right)\! =\! \frac{2\pi \delta_{\textrm{th}}}{N(t) |\dot{\sigma}(t)|}\! \triangleq\! \mathcal{K}_{\textrm{th}}(t).
	\end{equation}
	Analysis of the first derivative $g'(r)$ reveals that $g(r)$ is a unimodal function peaking at $r=\sigma(t)$. Consequently, for any feasible threshold $\mathcal{K}_{\textrm{th}}(t) < g(\sigma(t))$, the equation possesses two real roots. To ensure the activation zone fully encompasses the expanding swarm, $R_{\textrm{flux}}(t)$ is defined as the supremum (outermost) root, as formulated in (22c).
	
	\textit{3) Phase-Dependent Synthesis:}
	Finally, applying the asymmetric control logic: during the expansion phase where $\dot{\sigma} > 0$, the network activates if either trigger is met, implying $R_{\textrm{active}} = \max \{ R_{\textrm{act}}, R_{\textrm{flux}} \}$. During the contraction phase where $\dot{\sigma} \le 0$, the flux trigger is disabled, and the boundary is solely governed by the hysteresis condition, yielding $R_{\textrm{active}} = R_{\textrm{hold}}$. 
	
	This completes the proof.
\end{proof}

\section{Proof of Theorem 1}\label{app:proof_theorem_cov} 

\begin{proof}
	The instantaneous CP is defined as $\mathsf{P}_{\mathrm{cov}}(r,t) = \mathbb{P}(\textrm{SINR} > \gamma \mid r,t)$. Given the activation radius $R_{\textrm{active}}(t)$ derived in Proposition 1, we analyze the coverage in two distinct spatial regions.
	
	\subsubsection{Outage Zone ($r > R_{\mathrm{active}}(t)$)}
	
	In this region, the density of active LAS modules is zero. Consequently, no serving BS is available, yielding $\mathsf{P}_{\mathrm{cov}}(r,t) = 0$.
	
	\subsubsection{Active Zone ($r \le R_{\mathrm{active}}(t)$)}
	
	Conditioned on the nearest serving BS distance $y$, we focus on the interference-limited regime, i.e., $I_{\textrm{agg}} \gg P_{\mathrm{N}}$. Under this assumption, the noise term is negligible, and the CP is given by:
	\begin{align} 
		\mathbb{P}(\textrm{SINR} > \gamma \mid y) &\approx \mathbb{E}_{I_{\textrm{agg}}}\left[ \exp\left( - \gamma y^\nu P_{\textrm{LAS}}^{-1} I_{\textrm{agg}} \right) \right] \notag \\
		&= \mathcal{L}_{I_{\textrm{agg}}}(s)\big|_{s = \gamma y^\nu P_{\textrm{LAS}}^{-1}}.
		\label{eq:laplace_def}
	\end{align}
	Utilizing the probability generating functional of the PPP [22] truncated within the annulus $[y, R_{\textrm{active}}]$, we arrive:
	\begin{align} 
		\mathcal{L}_{I_{\textrm{agg}}}(s) &= \exp\left( -2\pi\lambda_{\textrm{BS}} \int_{y}^{R_{\textrm{active}}} \frac{s v}{v^\nu + s} \, \mathrm{d}v \right).
		\label{eq:integral_setup}
	\end{align}
	Let $\kappa\triangleq 2/\nu$. We invoke the standard integral identity involving the Gaussian hypergeometric function $_2F_1(\cdot)$:
	\begin{equation} 
		\int \frac{v}{v^\nu + s} \, \mathrm{d}v = \frac{v^2}{2s} \, {}_2F_1\left(1, \kappa; 1+\kappa; -v^\nu s^{-1}\right).
	\end{equation}
	Applying the integration limits $y$ and $R_{\textrm{active}}$ to \eqref{eq:integral_setup}, and substituting $s = \gamma y^\nu$, the integral term resolves to a closed-form interference exponent $\Psi(y, R_{\textrm{active}})$. Note that the finite integral $\int_y^R$ can be viewed as the difference between the infinite integral $\int_y^\infty$ and the tail integral $\int_R^\infty$. Thus:
	\begin{align} 
		\Psi &= \pi\lambda_{\textrm{BS}} s \Bigg( \frac{y^2}{s} {}_2F_1\left(1, \kappa; 1+\kappa; -\frac{y^\nu}{s}\right) \notag \\
		&\quad - \frac{R_{\textrm{active}}^2}{s} {}_2F_1\left(1, \kappa; 1+\kappa; -\frac{R_{\textrm{active}}^\nu}{s}\right) \Bigg) \notag \\
		&= \pi\lambda_{\textrm{BS}} \left( y^2 \mathcal{H}\left( -\frac{1}{\gamma} \right) - R_{\textrm{active}}^2 \mathcal{H}\left( -\frac{1}{\gamma} \eta^{-\nu} \right) \right),
		\label{eq:psi_derived}
	\end{align}
	where $\eta \triangleq y / R_{\textrm{active}}$ denotes the normalized radial distance, and $\mathcal{H}(z) \triangleq {}_2F_1(1, \kappa; 1+\kappa; z)$.
	
	Substituting \eqref{eq:psi_derived} back into the Laplace transform $\mathcal{L} = \exp(-\Psi)$ yields:
	\begin{equation} 
		\mathcal{L}_{I_{\textrm{agg}}}\! (\gamma y^\nu)\! =\! \exp\! \left( -\pi\lambda_{\textrm{BS}} y^2\!  \left( \mathcal{H}(-\gamma^{-1})\! - \eta^{-2} \mathcal{H}(-\gamma^{-1}\eta^{-\nu}) \right) \right)\! .
	\end{equation}
	De-conditioning over the PDF of the nearest neighbor distance $f_Y(y) = 2\pi\lambda_{\textrm{BS}} y e^{-\pi\lambda_{\textrm{BS}} y^2}$ for $y \in [0, R_{\textrm{active}}]$, we obtain:
	\begin{equation} 
		\mathsf{P}_{\mathrm{cov}}(r,t) = \int_{0}^{R_{\textrm{active}}} e^{-\Psi(y, R_{\textrm{active}})} f_Y(y) \, \mathrm{d}y.
	\end{equation}
	This completes the proof.
\end{proof}

\section{Proof of Theorem 2}\label{app:proof_theorem_ee} 

\begin{proof}
	For the network EE $\eta_{\mathrm{EE}} = \mathcal{T}_{\mathrm{total}} / E_{\mathrm{total}}$. We derive its numerator and denominator separately.
	
	\subsubsection{Total Served Traffic $\mathcal{T}_{\mathrm{total}}$}
	
	Substituting the instantaneous throughput density $\mathcal{R}(\mathbf{x}, t) = B\, \eta_{\rm SE} \lambda(\mathbf{x},t) \mathsf{P}_{\mathrm{cov}}(\mathbf{x},t)$ into the definition of total traffic, we integrate over the service area $\mathcal{A}$ and period $T_{\mathrm{period}}$:
	\begin{equation} 
		\mathcal{T}_{\mathrm{total}} = \int_0^{T_{\mathrm{period}}} \int_{\mathcal{A}} B\, \eta_{\rm SE} \lambda(\mathbf{x},t) \mathsf{P}_{\mathrm{cov}}(\mathbf{x},t) \, \mathrm{d}\mathbf{x} \, \mathrm{d}t.
	\end{equation}
	The service availability is strictly determined by the activation policy. For $r > R_{\mathrm{active}}(t)$, the GBSs are in sleep mode, yielding $\mathsf{P}_{\mathrm{cov}} = 0$. For the active region $r \le R_{\mathrm{active}}(t)$, the local CP varies spatially due to boundary effects as derived in Theorem 1.
	
	To facilitate a tractable system-level metric, we adopt the \textit{mean value approximation} for the spatial domain. We define the spatially averaged CP at time $t$ as:
	\begin{equation} 
		\bar{\mathsf{P}}_{\mathrm{cov}}(t) \triangleq \frac{1}{\pi R_{\mathrm{active}}^2(t)} \int_{0}^{R_{\mathrm{active}}(t)} \mathsf{P}_{\mathrm{cov}}(r,t) 2\pi r \, \mathrm{d}r.
	\end{equation}
	By decoupling the spatial average from the agent density distribution, the instantaneous network traffic $\mathcal{T}_{\mathrm{inst}}(t)$ is approximated as:
	\begin{align} 
		\mathcal{T}_{\mathrm{inst}}(t) &\approx \int_0^{R_{\mathrm{active}}(t)} B\, \eta_{\rm SE} \bar{\mathsf{P}}_{\mathrm{cov}}(t) \lambda(r,t) 2\pi r \, \mathrm{d}r \nonumber \\
		&= B\, \eta_{\rm SE} \bar{\mathsf{P}}_{\mathrm{cov}}(t) \underbrace{\int_0^{R_{\mathrm{active}}(t)} \lambda(r,t) 2\pi r \, \mathrm{d}r}_{\mathcal{M}(t)}.
		\label{eq:traffic_approx}
	\end{align}
	This approximation holds because the CP $\mathsf{P}_{\mathrm{cov}}(r,t)$ remains nearly constant within the agent-dense core region (as verified in Fig. 3), allowing it to be effectively factored out from the agent density integral.
	
	Evaluating the inner Gaussian integral:
	\begin{align} 
		\mathcal{M}(t) &= \frac{N(t)}{\sigma^2(t)} \int_0^{R_{\mathrm{active}}(t)} e^{-\frac{r^2}{2\sigma^2(t)}} r \, \mathrm{d}r \nonumber \\
		&= N(t) \left( 1 - \exp\left( - \frac{R_{\mathrm{active}}^2(t)}{2\sigma^2(t)} \right) \right).
	\end{align}
	Substituting $\mathcal{M}(t)$ back, we obtain $\mathcal{T}_{\mathrm{inst}}(t) = B\, \eta_{\rm SE} \bar{\mathsf{P}}_{\mathrm{cov}}(t) \mathcal{M}(t)$. Integrating $\mathcal{T}_{\mathrm{inst}}(t)$ over the period $T_{\mathrm{period}}$ yields the numerator of (25).
	
	\subsubsection{Total Energy Consumption $E_{\mathrm{total}}$}
	
	The instantaneous power consumption depends on the binary activation state $\alpha(\mathbf{x},t)$. We decompose the service area $\mathcal{A}$ into the AR $\mathcal{A}_{\mathrm{on}}(t)$ (a disk of radius $R_{\mathrm{active}}(t)$) and the sleep region $\mathcal{A}_{\mathrm{off}}(t)$:
	\begin{align} 
		P(t) &= \int_{\mathcal{A}} \lambda_{\mathrm{BS}} \left( P_{\mathrm{act}} \alpha(\mathbf{x},t) + P_{\mathrm{slp}} (1 - \alpha(\mathbf{x},t)) \right) \, \mathrm{d}\mathbf{x} \nonumber \\
		&= \lambda_{\mathrm{BS}} \left( P_{\mathrm{act}} |\mathcal{A}_{\mathrm{on}}(t)| + P_{\mathrm{slp}} (|\mathcal{A}| - |\mathcal{A}_{\mathrm{on}}(t)|) \right),
	\end{align}
	where $|\mathcal{A}_{\mathrm{on}}(t)| = \pi R_{\mathrm{active}}^2(t)$. Rearranging the terms:
	\begin{equation} 
		P(t) = \pi R_{\mathrm{active}}^2(t) \lambda_{\mathrm{BS}} (P_{\mathrm{act}} - P_{\mathrm{slp}}) + P_{\mathrm{base}},
	\end{equation}
	which matches the definition in (28) with $P_{\mathrm{base}} = \lambda_{\mathrm{BS}} |\mathcal{A}| P_{\mathrm{slp}}$. Integrating $P(t)$ over time yields $E_{\mathrm{total}}$.
	
	Combining $\mathcal{T}_{\mathrm{total}}$ and $E_{\mathrm{total}}$ completes the proof.
\end{proof}